\documentclass{article}

\usepackage[numbers,sort]{natbib}
\usepackage{times}

\usepackage{proceed2e}

\usepackage{amsmath,amsfonts,amssymb,amsthm}
\usepackage{accents}
\usepackage{graphicx,mathtools,verbatim}
\usepackage{url}
\usepackage{prettyref}
\usepackage{boxedminipage}
\usepackage[dvipsnames,usenames]{xcolor}
\usepackage{prettyref}
\usepackage[colorlinks=true,pdfpagemode=UseNone,urlcolor=Blue,linkcolor=RoyalBlue,citecolor=OliveGreen,pdfstartview=FitH]{hyperref}
\usepackage{bookmark}
\usepackage{wrapfig}
\usepackage{ifthen}
\usepackage{xcolor}
\usepackage{framed}
\usepackage{enumitem}
\usepackage{tabularx}
\usepackage{natbib}
\usepackage{flushend}
\setlist{noitemsep}

\makeatletter \def\dagfootnote{\xdef\@thefnmark{$\dagger$}\@footnotetext} \makeatother

\DeclareMathOperator*{\E}{\mathbb{E}}
\DeclareMathOperator{\conv}{conv}
\DeclareMathOperator{\aff}{aff}

\DeclareMathOperator{\im}{im}
\DeclareMathOperator{\ri}{ri}

\DeclareMathOperator{\dom}{dom}

\DeclareMathOperator{\linspan}{span}

\DeclareMathOperator*{\argmin}{argmin}

\DeclareMathOperator{\interior}{int}

\newcommand{\M}{\mathcal{M}}
\newcommand{\tM}{\tilde{\M}}
\newcommand{\inprod}{\cdot}
\newcommand{\wo}{\backslash}
\newcommand{\KL}{\text{KL}}

\newcommand{\hq}{\hat{q}}
\newcommand{\bq}{\bar{q}}
\newcommand{\hQ}{\hat{Q}}
\newcommand{\hnu}{\hat{\nu}}

\newcommand{\B}{\mathcal{B}}
\newcommand{\eps}{\varepsilon}

\newcommand{\ind}{I}
\newcommand{\setS}{\mathcal{S}}
\newcommand{\setB}{\mathcal{B}}
\newcommand{\setC}{\mathcal{C}}
\newcommand{\setK}{\mathcal{K}}
\newcommand{\abs}[1]{\lvert#1\rvert}

\newcommand{\tnu}{\tilde{\nu}}

\newcommand{\talpha}{\tilde{\alpha}}
\newcommand{\bnu}{\bar{\nu}}
\newcommand{\setL}{\mathcal{L}}
\newcommand{\witness}{\mathcal{K}}
\newcommand{\pinv}{p^{-1}}
\newcommand{\bB}{\bar{B}}

\newcommand{\ulambda}{\underaccent{\bar}{\lambda}}
\newcommand{\olambda}{\bar{\lambda}}

\renewcommand{\Re}{\mathbb{R}}

\newcommand{\Bracks}[1]{\left[#1\right]}

\newcommand{\Parens}[1]{\left(#1\right)}
\newcommand{\bigParens}[1]{\bigl(#1\bigr)}
\newcommand{\BigParens}[1]{\Bigl(#1\Bigr)}
\newcommand{\set}[1]{\{#1\}}
\newcommand{\norm}[1]{\lVert#1\rVert}

\newcommand{\Set}[1]{\left\{#1\right\}}

\newcommand{\Eq}[1]{Eq.~\eqref{eq:#1}}
\newcommand{\Thm}[1]{Theorem~\ref{thm:#1}}
\newcommand{\Fig}[1]{Figure~\ref{fig:#1}}
\newcommand{\Lem}[1]{Lemma~\ref{lem:#1}}
\newcommand{\Sec}[1]{Section~\ref{sec:#1}}
\newcommand{\App}[1]{Appendix~\ref{app:#1}}
\newcommand{\Cor}[1]{Corollary~\ref{cor:#1}}
\newcommand{\Prop}[1]{Proposition~\ref{prop:#1}}
\newcommand{\Def}[1]{Definition~\ref{def:#1}}
\numberwithin{equation}{section}

\theoremstyle{plain}
\newtheorem{theorem}{Theorem}[section]
\newtheorem{lemma}[theorem]{Lemma}
\newtheorem{corollary}[theorem]{Corollary}
\newtheorem{proposition}[theorem]{Proposition}

\theoremstyle{remark}

\theoremstyle{definition}
\newtheorem{definition}[theorem]{Definition}
\newtheorem{example}[theorem]{Example}

\renewcommand{\arraycolsep}{\tabcolsep}

\graphicspath{{./figs/}}

\title{Budget Constraints in Prediction Markets}

\author{Nikhil Devanur
\\{Microsoft Research}
\And
Miroslav Dud\'ik
\\{Microsoft Research}
\And
Zhiyi Huang
\\{University of Hong Kong}
\And
David M. Pennock
\\{Microsoft Research}}

\begin{document}

\maketitle

\begin{abstract}
We give a detailed characterization of optimal trades under budget constraints
in a prediction market with a cost-function-based automated market maker.
We study how the budget constraints of individual traders affect their ability to impact the market price.
As a concrete application of our characterization, we give sufficient conditions
for a property we call \emph{budget additivity}: two traders with budgets $B$ and $B'$ and the same beliefs would have a combined impact equal to a single trader with budget $B+B'$.
That way, even if a single trader cannot move the market much, a crowd of like-minded traders can have the same desired effect.
When the set of payoff vectors associated with outcomes, with coordinates corresponding to securities, is affinely independent,
we obtain that a generalization of the heavily-used logarithmic market scoring rule is budget additive, but the quadratic market scoring rule is not.
Our results may be used both descriptively, to understand if a particular market maker is affected by budget constraints or not, and prescriptively, as a recipe to construct markets.
\end{abstract}


\section{INTRODUCTION}
\label{sec:intro}

A prediction market is a central clearinghouse for people with differing opinions about the likelihood of an event---say Hillary Clinton to win the 2016 U.S. Presidential election---to trade monetary stakes in the outcome with one another. At equilibrium, the price to buy a contract paying \$1 if Clinton wins reflects a consensus of sorts on the probability of the event. At that price, and given the wagers already placed, no agent is willing to push the price further up or down. Prediction markets have a good track record of forecast accuracy in many domains \citep{GRWP10,WZ04}.

The design of {\em combinatorial} markets spanning multiple logically-related events raises many interesting questions. What information can be elicited---the full probability distribution, or specific properties of the distribution? What securities can the market allow traders to buy and sell?
How can the market support and ensure a variety of trades?
For example, in addition to the likelihood of Clinton winning the election, we may want to elicit information about the distribution of her electoral votes.\footnote{
A U.S. Presidential candidate receives a number of electoral votes between
0 and 538. The candidate who receives a plurality of electoral votes wins the election.}
If we create one security for each possible outcome between 0 and 538,
each paying \$1 iff Clinton gets exactly that many electoral votes, the market is called \emph{complete}, allowing us to elicit a full probability distribution. Alternatively, if we create just two securities, one paying out \$$x$
if Clinton wins $x$ electoral votes, and the other paying out \$$x^2$, we cannot elicit a full distribution, but we can still elicit the mean and variance of the number of electoral votes.

When agents are constrained in how much they can trade only by risk aversion, prediction market prices can be interpreted as a weighted average of traders' beliefs \citep{WZ06,BLP12}, a natural reflection of the ``wisdom of the crowd'' with a good empirical track record \citep{Jacobs95} and theoretical support \citep{BLP12}. However, when agents are budget constrained, discontinuities and idiosyncratic results can arise \citep{Manski06,Eisenberg59} that call into question whether the equilibrium prices can be trusted to reflect any kind of useful aggregation.

We consider prediction markets with an automated market maker \citep{AbernethyCV11,Hanson07,ChenP07} that maintains standing offers to trade every security at some price. Unlike a peer-to-peer exchange, all transactions route through the market maker.
The common market makers have bounded loss and are (myopically) incentive compatible:
the best (immediate) strategy is for a trader to move the market prices of all securities to equal his own belief.
The design of such an automated market maker boils down to choosing a convex
\emph{cost function}~\citep{AbernethyCV11}.
This amount of design freedom presents an opportunity to seek cost functions
that satisfy additional desiderata such as computational tractability~\citep{AbernethyCV11,DudikLaPe12}.

Most of the literature assumes either risk-neutral or risk-averse traders with unbounded budgets.
In this paper, we consider how agents with budget constraints trade in such markets, a practical reality in almost all prediction markets denominated in both real and virtual currencies.
Our results help with a systematic study
of the market's \emph{liquidity parameter},
or the parameter controlling the sensitivity of prices to trading volume. Setting the liquidity is a nearly universal practical concern and, at present, is more (black) art than science.
We adopt the notion of the ``natural budget constraint''
introduced by \citet{FortnowS12}: the agent is allowed only those trades for which
the maximum loss for any possible outcome does not exceed the budget.

The main contribution of this paper is  a rich, geometric characterization of the impact of budget constraints.
Price vectors, outcomes and trader beliefs are embedded in the space of the same dimension as the number of securities.
Outcome vectors enumerate security payoffs; belief vectors enumerate the traders' expectations of payoffs.
We consider, for a fixed belief, the locus of the resulting price vectors of an optimal trade as a function of the budget.
We show that the price vector moves in the convex hull of the belief and the set of tight outcomes, in a direction that is perpendicular to the set of tight outcomes.
We also introduce the concept of {\em budget additivity}:
two agents with budgets $B$ and $B'$ and the same beliefs
     have the same power to move the prices as a single agent with the same belief and budget $B+B'$.
     An absence of budget additivity points to an inefficiency in incorporating information from the traders.
We show that budget additivity is a non-trivial property by giving examples of market makers that do not satisfy budget additivity.
We give a set of sufficient conditions on the market maker and the set of securities offered which guarantee budget additivity.
Further, for two of the most commonly used market makers (the quadratic and logarithmic market scoring rules),
we show sufficient conditions on the set of securities that guarantee budget additivity.

Of greatest practical interest is the application of our results to markets consisting of several independent questions, with
each question priced according to a separate logarithmic market scoring rule. This setup constitutes a de facto industry standard,
and the companies that use (or used) it include Inkling Markets,\footnote{\url{inklingmarkets.com}} Consensus Point,\footnote{\url{www.consensuspoint.com}}
 Microsoft and Yahoo!~\citep{OthmanPeReSa10}. Our Theorems~\ref{thm:log:partition}
and~\ref{thm:direct:sum} show that these markets are budget additive.

Previously, \citet{FortnowS12} considered a different question: do budget-constrained bidders always move the market prices
in the direction of their beliefs? They showed that the answer to this is no: there always exist market prices, beliefs and budgets
such that the direction of price movement is not towards the belief. We give a richer characterization of how the market prices move in the
presence of budget constraints, by charting the path the prices take with increasing budgets.
The impossibility result of \citet{FortnowS12} can be easily derived from our characterization
(see \App{FortnowS12}).\footnote{%
The full version of this paper on arXiv includes the appendix.}

A designer of a prediction market has a lot of freedom but little guidance,
and our results can be used both descriptively and prescriptively.
As a descriptive tool, our results enable us to analyze commonly used market makers and understand if budget constraints
hamper information aggregation in these markets. As a prescriptive tool, our results can be used to construct markets that are budget additive.
In particular, we speculate that budget additivity simplifies the choice
of the liquidity parameter in the markets, because it allows considering trader budgets in aggregate.


\begin{figure*}
\vspace{-3ex}
\centering
~\hfill
\includegraphics[width=0.32\textwidth,trim=2pt 2pt 2pt 2pt,clip]{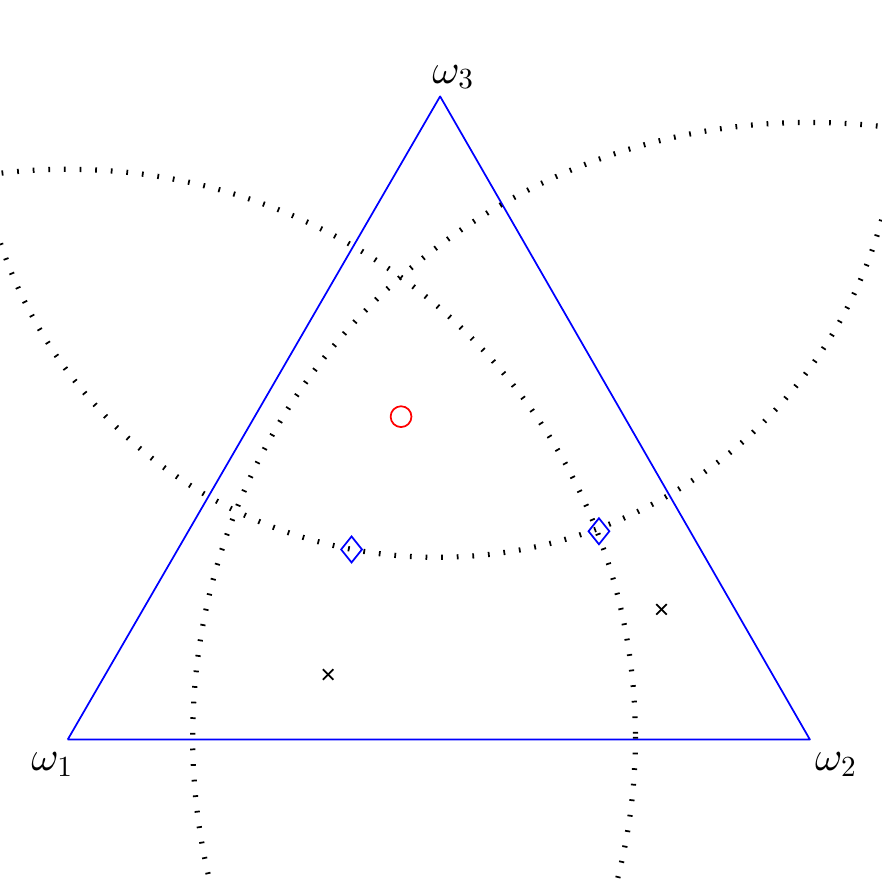}\hfill
\includegraphics[width=0.32\textwidth,trim=2pt 2pt 2pt 2pt,clip]{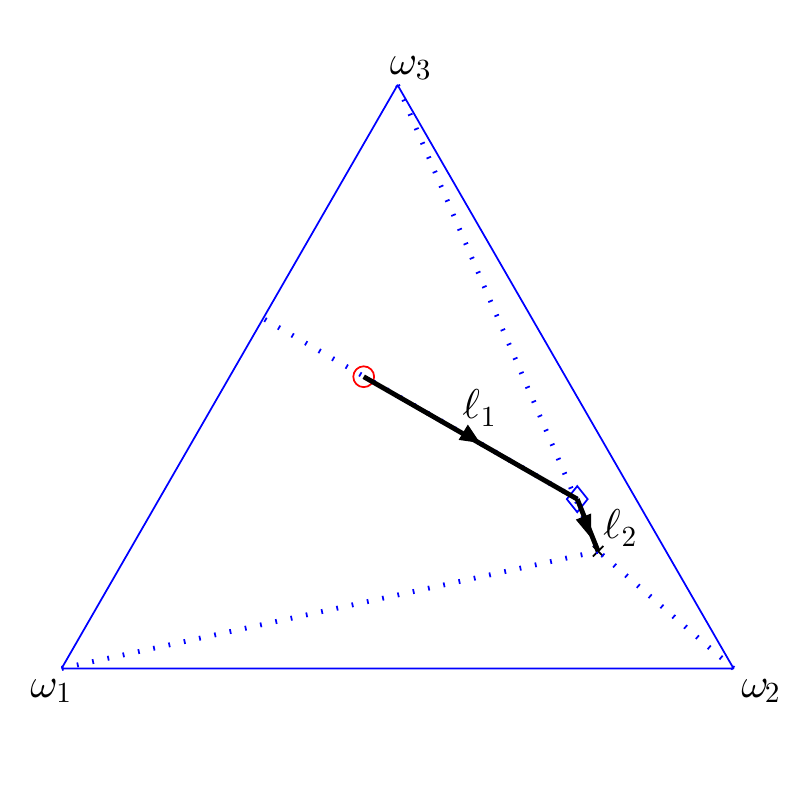}\hfill~
\vspace{-2ex}
\caption{\emph{Left:} $\circ$ ---current state, $\times$ ---belief, $\diamond$ ---optimal
action for a given belief and budget. Three circles bound the allowed final states for budget $0.1$. We plot optimal actions for two different beliefs.
\emph{Right:} A path from the initial state to the belief, consisting
of optimal actions for increasing budgets.}
\label{fig:budget}
\end{figure*}

\textbf{Proof overview and techniques.}
%
Our analysis borrows heavily from techniques in convex analysis and builds on the
 notion of Bregman divergence. We use the special case of Euclidean distance (corresponding to a quadratic market scoring rule) to form our geometric intuition which we then extend to arbitrary Bregman divergences.
For the sake of an example, consider a complete market over a finite set of outcomes, where the market prices lie in a simplex, exactly coinciding with the set of probability distributions over outcomes.
Every possible outcome imposes a constraint on the set of prices to which a trader can move the market, because the
trader is not allowed to exceed the budget if that outcome occurs.
The prices satisfying this constraint form a ball with the outcome at its center.
The set of feasible prices to which the trader can move the market is therefore the intersection of these balls (see \Fig{budget}).

The key structural result we obtain is the chart of the price movement. Suppose that there is an infinite sequence of agents with infinitesimally small budgets all with the same belief. What is the path along which the prices move from some initial values? This is determined by the agents' belief and the set of budget constraints that are tight at any point, corresponding to the highest risk outcomes (outcomes with the highest potential loss). We show that the price vector can always be written as a convex combination of these highest risk outcomes and the agents' belief.
Further,  the market prices move in a direction that is perpendicular to the affine space of these outcomes.

The agents' belief partitions the simplex interior into regions, where each region is the interior of the convex hull of the agent belief and a particular subset of outcomes. For a region that is full-dimensional, every interior point can be uniquely written as a convex combination of the agent belief and all except one outcome. Assume that the current price vector lies in this region. In the anticipation of the further development, we call this outcome \emph{profitable} and others \emph{risky}. Motivated by the characterization above, we move perpendicular to the risky outcomes in the direction towards the agents' belief. As a result, we increase the risk of risky outcomes (equally for all outcomes), while getting closer to the one profitable outcome (and hence increasing its profit). The characterization then guarantees that the prices along this path are indeed those chosen by traders at increasing budgets, because the risky outcomes yield tight constraints.

We would like the same to be true for the lower dimensional regions as well; that is, for the set of tight constraints
to be exactly the corresponding set of outcomes defining the convex hull.
In fact, this property is sufficient to guarantee budget additivity.
The markets for which
the tight constraints are exactly the minimal set of outcomes that define the region  the price lies in are  budget additive. (We conjecture that the converse holds as well.)
The entire path is then as follows: w.l.o.g.\ you start at a full-dimensional region, move along the perpendicular until you hit the boundary of the region and you are in a lower-dimensional region, move along the perpendicular in this lower-dimensional region, and so on until you reach the belief
(see \Fig{budget}).
The set of tight constraints is monotonically decreasing. We show that such markets are characterized by a certain \emph{acute angles assumption} on the set of possible outcomes. Loosely speaking, this assumption guarantees that outcomes outside the minimal set behave as the profitable outcome in the above example.

\textbf{Other related work.}
There is a rich literature on scoring rules and prediction markets.
Two of the most studied scoring rules are the quadratic scoring rule \citep{Brier} and
the logarithmic market scoring rule \citep{Hanson07}.
We consider cost-function-based prediction markets \citep{Hanson03,ChenP07}, a fully general class
under reasonable assumptions~\citep{AbernethyCV11,ChenV10}.
Their equivalence with proper scoring rules has been implicitly noted by
\citet{GneitingR07}.
%
%
Several authors have studied relationships between utility functions and price dynamics in prediction markets, drawing a parallel to online learning
\citep{ChenV10,BLP12,FrongilloDeRe12}. Our analysis touches on the problem of setting the market maker's
liquidity parameter~\citep{OthmanPeReSa10,LiWo13}, which determines how (in)sensitive prices are to trading volume. With budget additivity, the market designer can optimize liquidity according to aggregate budgets, without worrying about how budgets are partitioned among traders.

\section{PRELIMINARIES}

\textbf{Securities and payoffs.}
Consider a probability space with a finite set of outcomes $\Omega\subseteq\Re^n$.
A \emph{security} is a financial instrument whose payoff depends on the realization of an outcome in $\Omega$.
In other words, the payoff of a security is a random variable of the probability space. We consider trading with $n$ securities corresponding to $n$ coordinates
of the outcomes $\omega\in\Omega$.
A security can be traded before the realization is observed with the intention that the price of a security
serves as a prediction for the expected payoff, i.e., the expected value of the
 corresponding coordinate.

\textbf{Cost function, prices and utilities.}
An {\em automated market maker} always offers to trade securities, for the right price. In fact the price vector is the current prediction of
the market maker for the expectation of $\omega$.
A cost function based market maker is based on a differentiable convex \emph{cost function}, $C: \Re^n \to \Re$. It is a scalar function of an $n$-dimensional
vector $q\in \Re^n$ representing the {\em number of outstanding shares}\footnote{We allow trading fractions of a security. Negative values correspond to short-selling.} for our $n$ securities.
We also refer to $q$ as the {\em state}  of the market.

The vector of \emph{instantaneous prices} of the securities is simply the gradient of $C$ at $q$, denoted by
$
   p(q) \coloneqq \nabla C(q).
$
The prices of securities change continuously as the securities are traded, so it is useful to consider the cost of trading
a given quantity of securities. The cost of buying $\delta \in \Re^n$ units of securities (where a negative value corresponds to selling) is determined by the path integral
   $\int_\pi p(\bq)\inprod d\bq = C(q+\delta) - C(q)$,
where $\pi$ is any smooth curve from $q$ to $q+\delta$.

When the outcome $\omega$ is realized, the vector of $\delta$ units of securities pays off an amount of $\delta\cdot \omega$. Thus,
the realized utility of a trader whose trade $\delta$ moved the market state from $q$ to $q'=q+\delta$ is
\[
   U(q',\omega;q) \coloneqq (q'-q)\cdot\omega - C(q') + C(q)
\enspace.
\]
We make a standard assumption that the maximum achievable utility, which is also the maximum loss of the market maker, is bounded by a finite constant (in \Sec{convex}, we introduce a standard
approach to check this easily). Let $\M$ be the convex hull of the payoff vectors,
$
   \M\coloneqq\conv(\Omega)$.
It is easy to see that $\M$ contains exactly the vectors $\mu\in\Re^n$ which can be realized
as expected payoffs $\E[\omega]$ for some probability distribution over $\Omega$. For a trader
who believes that $\E[\omega]=\mu$, the expected utility takes form
\[
  U(q',\mu;q)\coloneqq\E\Bracks{U(q',\omega;q)} = (q'-q)\inprod\mu - C(q') + C(q)
\enspace.
\]

Throughout, we consider a single {\em myopic}  trader who trades as if he were the last to trade.
A key property satisfied by expected utility is \emph{path independence}: for any $q,\bq,q'\in\Re^n$,
  $U(q',\mu;\bq) + U(\bq,\mu;q) = U(q',\mu;q)$,
that is, risk-neutral traders have no incentive to split their trades. For a risk-neutral trader,
$q'\in\Re^n$ is an optimal action if and only if
    $\mu = \nabla C(q') = p(q')$
(this follows from the first-order optimality conditions). In other words, the trader is incentivized
to move the market to the prices corresponding to his belief as long as such prices exist.
In general, there may be multiple states yielding the same prices, so the inverse map $\pinv(\mu)$
returns a \emph{set}, which can be empty if no state yields the price vector $\mu$.


Commonly-used cost functions include the quadratic cost, logarithmic market-scoring rule (LMSR) and
the log-partition function. They are described in detail in \App{scoring:rules}.
The quadratic cost is defined by $C(q) = \frac{1}{2} \|q\|_2^2$ and $p(q)=q$.
Log-partition function is defined as
$C(q) = \ln (\sum_{\omega \in \Omega} e^{q \inprod \omega})$.
It subsumes LMSR as a special case for the \emph{complete market} with
the outcomes corresponding to vertices of the simplex. The
prices under log-partition cost
correspond to the expected value of $\omega$ under
the distribution $P_q(\omega)=e^{q\inprod\omega-C(q)}$
over $\Omega$, i.e., $p(q)=\E_{P_q}[\omega]$.


\textbf{Budget constraints.}
Trading in prediction markets needs an investment of capital. It is possible that an agent loses money on the trade,
in particular $U(q',\omega;q)$ could be negative for some $\omega$. One restriction on how an agent trades could be that he is unable to sustain a big loss, due to a budget constraint.
We consider the notion of {\em natural budget constraint} defined by \citet{FortnowS12} which states that
the loss of the agent is at most his budget,  for all $\omega \in \Omega$. Given a starting market state
$q_0$ and a budget of $B\ge 0$, a trader with the belief $\mu\in\M$ then solves the problem:
\newcommand{\CP}{Convex Program~\eqref{convex:program}}%
\begin{equation}
\label{convex:program}
\begin{aligned}
&
  \max_{q\in\Re^n} ~ U(q,\mu; q_0)
\\
&
  \text{\ s.t.\ }
  U(q,\omega; q_0) \ge -B\quad\forall\omega\in\Omega
\enspace.
\end{aligned}
\end{equation}
For quadratic costs, each constraint corresponds to a sphere with one of the outcomes at its center,
so the feasible region is an intersection of these spheres.
We will later see that this generalizes to an intersection of balls w.r.t.\@ a Bregman divergence for general costs.

In general, there may be multiple $q$ optimizing this objective. In the following definition we
introduce notation for various solution sets we will be analyzing.
The belief $\mu$ is fixed throughout most of the discussion, so we suppress the dependence on $\mu$.
\begin{definition}[Solution sets]
Let $\hQ(B;q_0)$ denote the set of solutions of \CP\ for a fixed initial state and budget.
Let $\hQ(q_0)=\bigcup_{B\ge 0}\hQ(B;q_0)$
denote the set of solutions of \eqref{convex:program} for a fixed initial state
across all budgets. Let $\hQ(\nu;q_0)=\pinv(\nu)\cap\hQ(q_0)$
denote the set of states $q$ that optimize \eqref{convex:program} for some budget $B$ and yield the market price vector
$\nu$.
\end{definition}
%
%
%
The next theorem shows that solutions for a fixed initial state and budget always yield the same
price vector. It is proved in \App{budget:constraints}.
\begin{theorem}
\label{thm:budget:constraints}
If $q,q'\in\hQ(B;q_0)$, then $p(q)=p(q')$.
\end{theorem}

\textbf{Geometry of linear spaces.}
We finish this section by reviewing a few standard geometric definitions we use in next sections.
Let $X\subseteq\Re^n$.
Then $\aff(X)$ denotes the \emph{affine hull} of the set $X$ (i.e., the smallest affine space including $X$).
We write $X^\perp$ to denote the \emph{orthogonal complement} of $X$: $X^\perp\coloneqq\set{u\in\Re^n:\:u\inprod(x'-x)=0\text{ for all }x,x'\in X}$. We
use the convention $\emptyset^\perp=\Re^n$.
A set $\setK\in\Re^n$ is called a \emph{cone} if it is closed under multiplication by positive scalars.
If a cone is convex, it is also closed under addition.
Since $\Omega$ is finite, the realizable set $\M=\conv(\Omega)$ is a polytope. Its boundary can
be decomposed into \emph{faces}. More precisely,
$X\subseteq\Omega$, $X\ne\emptyset$, forms a \emph{face} of $\M$ if $X$
is the set of maximizers over $\Omega$ of some linear function.\footnote{%
  Strictly speaking, this is the definition of an \emph{exposed face}, but all faces of a polytope are exposed, so
  the distinction does not matter here. The exposed face is typically defined to be $\conv(X)$, but in this paper, it is more convenient to work with $X$ directly.}
We also view $X=\emptyset$ as a face of $\M$. With this definition, for any two faces $X$, $X'$,
also their intersection $X\cap X'$ is a face.

\section{CHARACTERIZING SOLUTION SETS}



We start with the optimality (KKT) conditions for the Convex Program (\ref{convex:program}), as characterized by the next lemma.
One of the key conditions is that the solution prices must be in the convex hull of the belief $\mu$
and all the $\omega$'s for which the budget constraints are {\em tight}.
The set of tight constraints is always a face of the polytope $\M$. We allow an empty set as a face,
which corresponds to the case when none of the constraints are tight and the solution prices coincide with $\mu$. The proof follows by analyzing KKT conditions (see \App{proof:KKT} of the full version for details).

\begin{lemma}[KKT lemma]
\label{lem:KKT:initial}
Let $q_0\in\Re^n$. Then $q\in\hQ(B;q_0)$ if and only if there exists a face $X\subseteq\Omega$ such that the following conditions hold:
\begin{enumerate}[label=\textup{(\alph*)},topsep=0pt]
\item\label{kkt:tight}
  $U(q,x;q_0)=U(q,x';q_0)$, or equivalently\\
  \hphantom{~~~}$(q-q_0)\inprod(x'-x)=0$, for all $x,x'\in X$
\item\label{kkt:nontight}
  $U(q,\omega;q_0)\ge U(q,x;q_0)$, or equivalently\\
  \hphantom{~~~}$(q-q_0)\inprod(\omega-x)\ge 0$, for all $x\in X$, $\omega\in\Omega\wo X$
\item\label{kkt:conv}
  $p(q)\in\conv(X\cup\set{\mu})$
\item\label{kkt:budget}
$B=-U(q,x;q_0)$ for all $x\in X$ if $X\ne\emptyset$, or\\
$B\ge\max_{\omega\in\Omega}[-U(q,\omega;q_0)]$ if $X=\emptyset$
\end{enumerate}
where conditions \ref{kkt:tight} and \ref{kkt:nontight} hold vacuously for $X=\emptyset$.
\end{lemma}

The condition \ref{kkt:tight} requires that
$q-q_0$ be orthogonal to the active set $X$. The set of points
satisfying conditions \ref{kkt:tight} and $\ref{kkt:conv}$
will be called the Bregman perpendicular and will be defined in the
next section. The condition \ref{kkt:nontight} is a statement about
acuteness of the angle between $q-q_0$ (the perpendicular)
and the outcomes. It will be the basis of our \emph{acute angles}
assumption. The condition \ref{kkt:budget} just states how the budget
is related to the active set $X$.

\textbf{Witness cones and minimal faces.}
We now introduce some notation to help us state reinterpretations of the conditions in \Lem{KKT:initial}. First of all, given a face $X$,  what is the set of $q$'s that satisfy conditions \ref{kkt:tight} and \ref{kkt:nontight}?
This is captured by what we call the \emph{witness cone}.
\begin{definition}
The \emph{witness cone} for a face $X\subseteq\Omega$ is defined as
$\witness(X) \coloneqq \set{u\in\Re^n:\: u\inprod(\omega-x)\ge0\text{ for all }x\in X,\omega\in\Omega}$ if $X\ne\emptyset$, and $\witness(X) \coloneqq \Re^n$ if $X = \emptyset$.
\end{definition}
The following two properties of witness cones are immediate from the definition:
\begin{itemize}[topsep=0pt]
\item \emph{Anti-monotonicity:}
if $X\subseteq X'$, then $\witness(X)\supseteq\witness(X')$.
\item \emph{Orthogonality:}
$\witness(X)\subseteq X^\perp$.
\end{itemize}
A state $q$ satisfies  conditions \ref{kkt:tight} and \ref{kkt:nontight} for a given face $X$ if and only if $q-q_0 \in \witness(X)$. Now given a state $q$, consider the set of faces that could satisfy condition \ref{kkt:conv}.
This set has a useful structure, namely that there is a unique minimal face
(proved in \App{proof:KKT} of the full version).
\begin{definition}
Given a price vector $\nu\in\M$, the \emph{minimal face} for $\nu$ is the minimal face $X$ (under inclusion)
s.t.\ $\nu\in\conv(X\cup\set{\mu})$. The minimal face for $\nu$ is denoted as $X_\nu$.
\end{definition}
With the existence of a minimal face and the anti-monotonicity of the witness sets,
it follows that if $q$ and $X$ satisfy conditions  \ref{kkt:tight},  \ref{kkt:nontight} and  \ref{kkt:conv}, then so do $q$ and $X_{p(q)}$.
Thus we obtain the following version of \Lem{KKT:initial} (proved in \App{proof:KKT} of the full version).
\begin{theorem}[Characterization of Solution Sets]
\label{thm:KKT:final}
$q\in\hQ(q_0)$ if and only if $q\in[q_0+\witness(X_{p(q)})]$.
\end{theorem}

Using \Thm{KKT:final}, we immediately obtain a characterization of when a price vector $\nu$ could be the price vector of an optimal solution to \eqref{convex:program}.
%
\begin{corollary}
\label{cor:budget:nu}
$\hQ(\nu;q_0) = \pinv(\nu)\cap[q_0+\witness(X_\nu)]$. In particular, $\nu$ is the price vector of an optimal solution to \eqref{convex:program} if and only if $\pinv(\nu)\cap[q_0+\witness(X_\nu)] \neq \emptyset.$
\end{corollary}

\begin{figure*}
\centering
~\hfill
\includegraphics[width=0.3\textwidth,trim=2pt 2pt 2pt 2pt,clip]{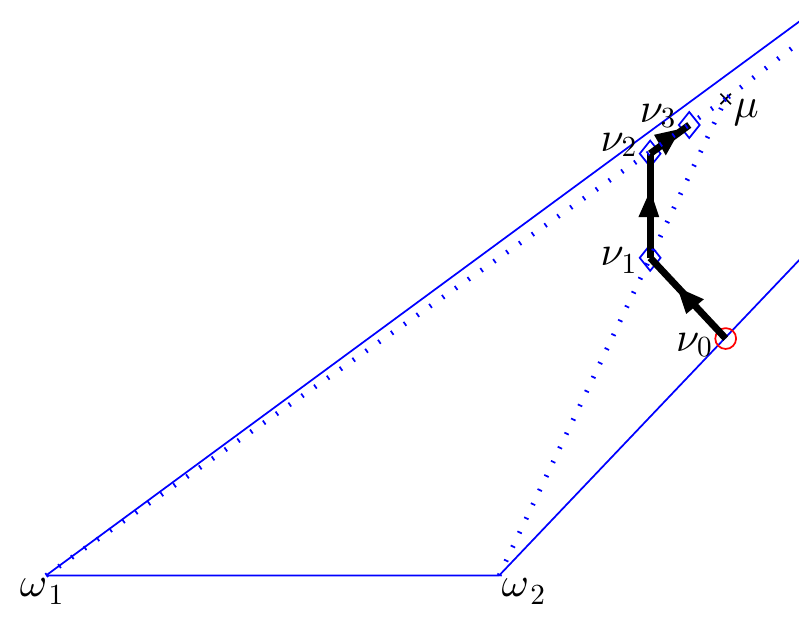}\hfill
\includegraphics[width=0.35\textwidth,trim=2pt 2pt 2pt 2pt,clip]{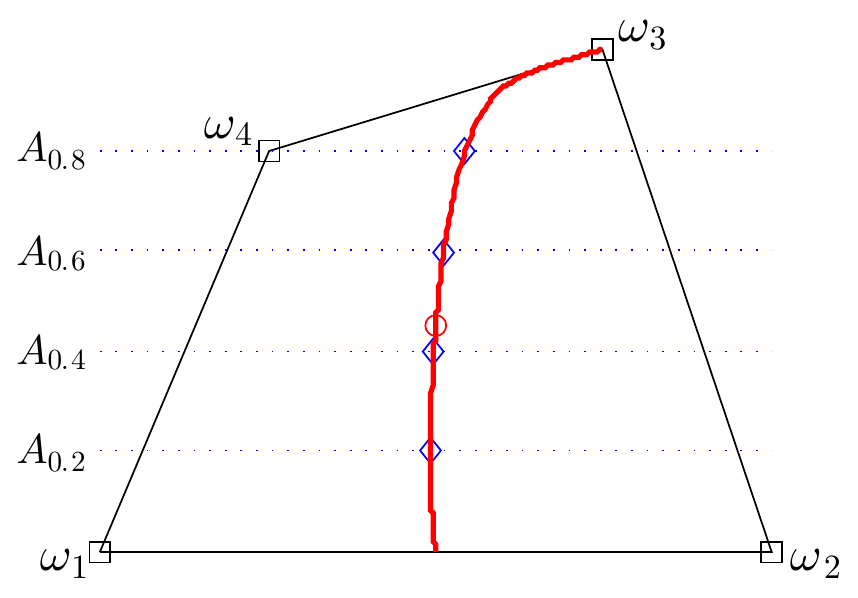}\hfill~
\caption{\emph{Left:} An example of non-additive budgets when
payoffs form obtuse angles (see Example~\ref{example:obtuse:short} and its
extended version Example~\ref{example:obtuse:long} in the full version).
\emph{Right:} An examples of a non-linear perpendicular for the log-partition cost.}
\label{fig:obtuse:perp}
\end{figure*}


We now study an example using the above characterization.
More examples can be found in \App{additivity:examples} of the full version.
\begin{example}[Quadratic cost on an obtuse triangle; see Example~\ref{example:obtuse:long} in the full version for details]
\label{example:obtuse:short}
Consider the following outcome space, belief, and the
sequence of market states (depicted in \Fig{obtuse:perp}):
\[
\renewcommand{\arraycolsep}{0pt}
\renewcommand{\arraystretch}{1.1}
\begin{array} {rl@{\qquad}rl}
  \omega_1 &{}= (0.0,\,0.0)
&
  q_0 &{}= \nu_0 = \frac{11}{14}\omega_2 + \frac{3}{14}\omega_3
\\
  \omega_2 &{}= (1.8,\,0.0)
&
  q_1 &{}= \nu_1 = \frac13\omega_2 + \frac23\mu
\\
  \omega_3 &{}= (6.0,\,4.2)
&
  q_2 &{}= \nu_2 = \frac19\omega_1 + \frac89\mu
\\
  \mu &{}= q_\mu = (2.7,\,1.8)
&
  q_3 &{}= \nu_3 \approx \frac{1}{19}\omega_1 + \frac{18}{19}\mu
\end{array}
\]
Using the KKT lemma, we can show for $j=1,2,3$,
that $q_j=\nu_j$ is an optimal action at $q_{j-1}=\nu_{j-1}$ under belief $\mu$, with the
corresponding budgets as:
\[
\renewcommand{\arraycolsep}{4pt}
\begin{array}{c|c@{~}c@{~}c|c}
& \omega_1
& \omega_2
& \omega_3
&
\\
\hline
U(q_1,\cdot;\,q_0)
& 0.45 & -0.09 & -0.09
& B_{01}=0.09
\\
U(q_2,\cdot;\,q_1)
& -0.56 & -0.56 & 1.12
& B_{12}=0.56
\\
U(q_3,\cdot;\,q_2)
& -0.565 & -0.28\dots & 0.82\dots
& B_{23}=0.565
\\
U(q_\mu,\cdot;\,q_0)
& -1.215 & -1.215 & 2.565
& B_{0\mu}=1.215
\end{array}
\]
The above table also shows that the budget $B_{0\mu}=1.215$ suffices to move directly
from $q_0$ to $q_\mu$. However, note that the sum
$ B_{01}+B_{12}+B_{23}=1.215=B_{0\mu},$
but $\nu_3\ne\mu$, i.e., after the sequence of optimal actions with budgets $B_{01}$,
$B_{12}$, and $B_{23}$, the market is still not at the belief shared by all agents,
even though with the budget $B_{0\mu}$, it would have reached it.
\end{example}

\paragraph{Budget additivity.}
The above example suggests that multiple traders with the same belief may have less power in moving the market state towards their belief compared to a single trader with the same belief and the combined budget.
Since prediction markets aim to efficiently aggregate information from agents, it is natural to ask under what conditions multiple traders with the same beliefs do have a combined impact equal to a single trader with the combined budget.

Next, we formally define this property as {\em budget additivity}.
We then define the Euclidean version of the acute angles condition that we show is sufficient for budget additivity.


\begin{definition}[Budget additivity]
\label{def:budget:additivity}
We say that a prediction market is budget additive on $\M'\subseteq\M$
if for all beliefs $\mu\in\M'$ and
all initial states $q_0\in\pinv(\M')$ the following holds: For any budgets $B,B'\ge 0$ and
any sequence of solutions $q\in\hQ(B;q_0)$ and $q'\in\hQ(B';q)$, we have $p(q),p(q')\in\M'$ and $q'\in\hQ(B+B';q_0)$.
\end{definition}

In other words, the market is budget additive if the sequence of optimal actions of
two agents with the same belief and budgets $B$ and $B'$
is also an optimal action of a single agent with the same belief and a larger budget $B+B'$.
Thanks to \Thm{budget:constraints} we then also obtain that the price vector following
the sequence of optimal actions by the two agents is the same as the price vector after the optimal action by an agent with the combined budget (all with the same beliefs).
%
%

We now state the acute angles assumption for the Euclidean case, to give an intuition. Our acute angles
assumption (Definition~\ref{def:acute}) is a generalization of this. We later show that the acute angles property is sufficient for budget additivity (Theorem~\ref{thm:main}).
\begin{definition}\label{def:Eacute}
We say that the \emph{Euclidean acute angles} hold for a face $X$, if the angle
between any point $\bnu\in\M$, its projection on the affine hull of $X$
and any payoff $\omega\in\Omega$ is non-obtuse (the angle is measured at the projection).
\end{definition}

Based on the above example, one may hypothesize that the obtuse angles
are to blame for the lack of budget additivity. In the following sections we will show
that this is indeed the case, but that the notion of obtuse/acute
angles depends on the Bregman divergence. In particular, the above
example would have been budget-additive if we used the log-partition cost instead of
the quadratic cost.

\section{BREGMAN DIVERGENCE AND PERPENDICULARS}
\label{sec:convex}


We will see next that the utility function $U$ can be written as the difference of two terms measuring the distance between the belief and the market state before and after the trade.
This distance measure is the mixed \emph{Bregman divergence}.\footnote{%
 Our notion of Bregman divergence is more general than typically assumed in the literature.}
To define the Bregman divergence, first let
$C^*:\Re^n\to(-\infty,\infty]$ be the convex conjugate of $C$ defined as
   $C^*(\nu) \coloneqq \sup_{q'\in\Re^n} \Bracks{q'\inprod\nu - C(q')}.$
Since $C^*$ is a supremum of linear functions, it is convex lower semi-continuous.
Up to a constant,
it characterizes the
maximum achievable utility on
an outcome $\omega$ for a fixed initial state $q$ as
  $\sup_{q'\in\Re^n} U(q',\omega;q) = C^*(\omega) + \big[C(q) - q\inprod\omega\big]$.
The term in the brackets is always finite, but $C^*$ might be positive infinite.
We make a standard assumption that $C^*(\omega)<\infty$ for all $\omega\in\Omega$, i.e.,
that the maximum achievable utility, which is also the maximum loss of the market maker, is bounded by a finite constant.
By convexity, this implies that
$C^*(\mu)<\infty$ for all $\mu\in\M$.
The \emph{Bregman divergence} derived from $C$ is a function $D:\Re^n\times\Re^n\to(-\infty,\infty]$ measuring the
maximum expected utility under belief $\mu$ at a state $q$
\[
\textstyle
  D(q,\mu) \coloneqq C(q) + C^*(\mu)-q\inprod\mu = \sup_{q'\in\Re^n} U(q',\mu;q)
\enspace.
\]
From the convexity of $C$ and $C^*$ and the definition of $C^*$, it is clear that: (i) $D$ is convex and lower semi-continuous in each argument separately;
(ii) $D$ is non-negative; and (iii) $D$ is zero iff $p(q)=\nabla C(q)=\mu$.
By the bounded loss assumption, Bregman divergence is finite on $\mu\in\M$. For $\mu\in\M$,
we can write
\begin{equation}
\label{eq:U}
  U(q',\mu;q) = D(q,\mu) - D(q',\mu)
\enspace.
\end{equation}
Thus, maximizing the expected utility is the same as minimizing the Bregman divergence
between the state $q'$ and the belief $\mu$.
From \Eq{U} it is also clear that each constraint in (\ref{convex:program})  is equivalent to $D(q,\omega)\le D(q_0,\omega)+B$, and the geometric interpretation is that
the agent seeks to find the state closest to his belief, within the intersection of Bregman balls

For the quadratic cost, we have $C^*(\nu)=\frac12\norm{\nu}^2$ and $D(q,\nu)=\frac12\norm{q-\nu}^2$,
i.e., the Bregman divergence coincides with the Euclidean distance squared.
For log-partition cost, we have
$C^*(\nu)=\sum_{\omega\in\Omega} P_\nu(\omega)\ln P_\nu(\omega)$ where $P_\nu$
is the distribution maximizing entropy among $P$ satisfying $\E_P[\omega]=\nu$. The
Bregman divergence is the KL-divergence between $P_q$ and $P_\nu$: $D(q,\nu)=\KL(P_\nu\Vert P_q)$.

\textbf{Convex analysis.}
We overview a few standard definitions and results from convex analysis.
For $X\subseteq\Re^n$, we write $\ri X$ for
the \emph{relative interior} of $X$ (i.e., the interior relative to the affine hull).
For a convex function $F:\Re^n\to(-\infty,\infty]$, we define its \emph{effective domain} as
$\dom F\coloneqq\set{u\in\Re^n:\:F(u)<\infty}$ (i.e., the set of points where it is finite). The \emph{subdifferential} of $F$ at a point $u$ is the set
   $\partial F(u) \coloneqq \set{v\in\Re^n:\: F(u')\ge F(u)+(u'-u)\inprod v\text{ for all }u'\in\Re^n}$.
We say that $F$ is \emph{subdifferentiable} at $u$ if $\partial F(u)\ne\emptyset$.
A standard result of convex analysis states that $F$ is always
subdifferentiable on a superset of $\ri\dom F$.
If $F$ is not only convex, but also lower semi-continuous, then $\partial F$ and
$\partial F^*$ are inverses in the sense that $v\in\partial F(u)$ iff $u\in\partial F^*(v)$.
If $F$ is \emph{differentiable} everywhere on $\Re^n$, then $F^*$ is strictly convex
on $\ri\dom F^*$.

Let $\im p\coloneqq\set{p(q):\:q\in\Re^n}$ denote the set of prices that can be expressed by market
states. The implications for our setting are that: (i) $C^*$ is subdifferentiable on $\im p$;
(ii) $\pinv(\nu)=\partial C^*(\nu)$ for all $\nu\in\Re^n$; (iii) all beliefs in $\ri\dom C^*$
can be expressed by some state $q$; (iv) $C^*$ is strictly convex on $\ri\dom C^*$, and similarly $D(q,\nu)$ is strictly convex on $\ri\dom C^*$
as a function of the second argument.
%
%

\textbf{Assumptions on the cost function.}
%
\begin{itemize}[topsep=0pt]
\item\emph{Convexity and differentiability on $\Re^n$.}
     $C$ is convex and differentiable on $\Re^n$.
\item\emph{Finite loss.}
     $\M\subseteq\dom C^*$, i.e., $C^*$ is finite on $\M$.
\item\emph{Inclusion of the relative interior.}
     $\ri\M\subseteq\ri\dom C^*$.
\end{itemize}
The first two assumptions are standard. The third assumption is a regularity condition that we require
in our results. Here we briefly discuss how it compares with the finite loss assumption.
While the two assumptions look similar, neither of them
implies the other. For example, if $\dom C^*$ is an $n$-dimensional simplex and $\M$ is one of its lower
dimensional faces,
which are lower dimensional simplices, then the finite loss assumption holds, but the inclusion assumption does not.
Similarly, for $n=1$ and $\M=[0,1]$, the inclusion assumption is satisfied by the conjugate $C^*(\nu)=1/\nu+1/(1-\nu)$ on $\nu\in(0,1)$ and $C^*(\nu)=\infty$ on $\nu\not\in(0,1)$, but this conjugate does not satisfy the finite loss assumption.

We do not view the inclusion assumption as very restrictive, since it is satisfied by many common
cost functions. For instance, it always holds when $C$ is constructed as in \cite{AbernethyCV11}, because their
construction guarantees $\dom C^*=\M$. However, the inclusion assumption might not hold for cost functions that allow arbitrage (e.g., \citep{DudikLaPe12}).

Our main result relies on strict convexity of $C^*$ on $\ri\dom C^*$, so some of our
statements will require that the market prices and beliefs lie in that set. The inclusion assumption above
guarantees that at the minimum
 $\ri\M\subseteq\ri\dom C^*$, but the boundary of $\M$ is not necessarily included.
To allow some generality beyond $\ri\M$, we define the set
\[
   \tM\coloneqq\begin{cases}
   \M&\text{if $\M\subseteq\ri\dom C^*$}
\\
   \ri\M&\text{otherwise.}
   \end{cases}
\]
In either case we obtain that
$\tM\subseteq\ri\dom C^*\subseteq\im p$, i.e., beliefs in $\tM$ can be expressed by some state $q$.
For the quadratic cost, $\tM=\M$. For the log-partition cost, $\tM=\ri\M$.

\textbf{Perpendiculars.}
We now define the notion of a Bregman perpendicular to an affine space.
This is a {\em constructive} definition. It plays a central role in the definition of
the acute angles assumption, and also in the proof of the main result (\Thm{main}).
We will see that the set of optimal price vectors for different budgets is a sequence of Bregman
perpendiculars. Naturally, perpendiculars are closely related to the conditions in Lemma \ref{lem:KKT:initial}; in particular to the set of $q$'s that satisfy conditions \ref{kkt:tight} and \ref{kkt:conv} for a given face $X$.

For quadratic costs, Bregman perpendiculars coincide with the usual Euclidean perpendiculars.
Consider an affine space and a point not in it. A projection of the point onto the space is the point in the space that is closest in Euclidean distance to the given point.
Now consider moving this affine space towards the projected point. The locus of the projection as we move the space is the perpendicular to the space through the given point.
We extend this definition to arbitrary Bregman divergences by defining the projection  using the corresponding Bregman divergence.

A Bregman perpendicular is determined by three geometric objects within the affine hull
$\aff(\dom C^*)$.
The first of these is an affine space, say $A_0 \subseteq \aff(\dom C^*)$.
The second is a point $a_1\in\aff(\dom C^*)\wo A_0$. The affine space
$A=\aff(A_0\cup\set{a_1})\subseteq\aff(\dom C^*)$
will be the ambient space that will contain the perpendicular.
Define parallel spaces to $A_0$ in $A$, for an arbitrary point $a_0 \in A_0$, as
    $A_\lambda\coloneqq A_0 + \lambda(a_1-a_0)$
for $\lambda\in\Re$. Note that the definition of $A_\lambda$ is independent of the choice of $a_0$.
The third geometric object is a market state $q\in\Re^n$ such that $p(q)\in A$. For technical reasons,
we will define a perpendicular at $q$ rather than a more natural notion, which would be at $p(q)$. Our
reason for switching into $q$-space is that inner products, defining optimality of the Bregman projection,
are between elements of $q$-space and $\nu$-space (the two spaces
coincide for Euclidean distance).
For all $\lambda\in\Re$ define a Bregman projection of $q$ onto $A_\lambda$ as
\[
    \nu_\lambda\coloneqq\argmin_{\nu\in A_\lambda} D(q,\nu)
\enspace.
\]
Since $D(q,\nu)$ is bounded from below and lower semi-continuous, the minimum is always attained
(but it may be equal to $\infty$).
If it is attained at more than one point, we choose an arbitrary minimizer. Whenever we can choose
$\nu_\lambda\in\ri\dom C^*$, this $\nu_\lambda$ must be the unique minimizer by strict
convexity of $D(q,\cdot)$ on $\ri\dom C^*$, and the minimum is finite.
We use these $\nu_\lambda$'s to define the perpendicular:

\begin{definition}
Given $A_0$, $a_1$ and $q$ as above, the $a_1$-\emph{perpendicular} to $A_0$ at $q$
is a map $\gamma:\lambda\mapsto\nu_\lambda$ defined over
$\lambda\in\Lambda\coloneqq\set{\lambda\in\Re:\nu_\lambda\in\ri\dom C^*}$. We call $\Lambda$
the domain of the perpendicular.
We define a total order on $\nu_{\lambda},\nu_{\lambda'}\in\im\gamma$ as
$\nu_{\lambda}\preceq\nu_{\lambda'}$ iff $\lambda\le\lambda'$.
\end{definition}

In \App{perp:cont} of the full version, we show that perpendiculars are continuous maps.
The name perpendicular is justified by the following proposition
which matches our Euclidean intuition that the perpendiculars can be obtained
by intersecting the ambient space $A$
with the affine space which passes through $q$ and is orthogonal to $A_0$.
It also shows that the perpendicular corresponds to the set of prices that satisfy conditions \ref{kkt:tight} and \ref{kkt:conv} with the convex hull relaxed to the affine hull (when $A_0$ is the affine hull of
face $X$, point $a_1$ coincides with $\mu$ and $q$ is the initial state).
Recall that for an arbitrary set $X\subseteq\Re^n$, its
orthogonal complement is defined as
$X^\perp\coloneqq\set{u:\: u\inprod(x'-x)=0\text{ for all }x,x'\in X}$.

\begin{proposition}
\label{prop:perp:dual}
Let $\gamma$ be the $a_1$-perpendicular to $A_0$ at $q$, and let $A=\aff(A_0\cup\set{a_1})$.
The following two statements are equivalent for any $\nu'\in\Re^n$:
\begin{itemize}[topsep=0pt]
\item[\textup{(i)}] $\nu'\in\im\gamma$
\item[\textup{(ii)}] $\nu'\in A\cap(\ri\dom C^*)$, $\pinv(\nu')\cap(q+A_0^\perp)\ne\emptyset$
\end{itemize}
\end{proposition}
\Prop{perp:dual} is proved in \App{perp} of the full version.
The perpendiculars  have the following closure property which is useful for showing budget additivity
(also proved in \App{perp} of the full version):
\begin{proposition}
\label{prop:perp:equiv}
Under the assumptions of \Prop{perp:dual},
$\gamma$ is also the $a_1$-perpendicular to $A_0$ at any $q'\in\pinv(\im\gamma)\cap(q+A_0^\perp)$.
\end{proposition}

\section{BUDGET ADDITIVITY}
\label{sec:main}
We now state the \emph{acute angles} property which links
the Bregman perpendicular and  \Cor{budget:nu},  and is sufficient for  budget additivity.
\begin{definition}
\label{def:acute}
We say that the \emph{acute angles} hold for a face $X$, if for every
$\mu$-perpendicular $\gamma$ to $X$ at $q$, such that $\mu\in\tM$ and
$q\in\pinv(\tM)$, the following holds:
If $\nu'\in\im\gamma$ and $\nu'\succeq p(q)$, then $\pinv(\nu')\cap[q+\witness(X)]\ne\emptyset$.
\end{definition}

The motivation for the name ``acute angles'' comes from the Euclidean distance case,
where this assumption is equivalent to Definition \ref{def:Eacute} (see \Prop{Eacute} in the full version).
The acute angles property is non-trivial and we have seen that without this property, budget
additivity need not hold; we conjecture that it is also a necessary condition.
After stating the main theorem, we analyze in more detail
when the acute angles are satisfied by the quadratic and log-partition costs.

We now state the main result, that the acute angles are sufficient for budget additivity:
\begin{theorem}[Sufficient conditions for budget additivity]
\label{thm:main}
If acute angles hold for every face $X\subseteq\Omega$,
then the prediction market is budget additive on $\tM$.
\end{theorem}

\textbf{Sufficient conditions for acute angles.}
%
We next give the sufficient conditions when
the acute angles hold for the quadratic and log-partition cost functions.
We also show that the acute angles hold for all one-dimensional outcome spaces, and
that they are preserved by taking \emph{direct sums} of markets.
Recall that a set $\setK\in\Re^n$ is called a cone if it is closed under multiplication by positive scalars.
A cone is called \emph{acute}, if $x\inprod y\ge 0$ for all $x,y\in\setK$. An \emph{affine cone} with the
vertex $a_0$ is a set $\setK'$ of the form $a_0+\setK$ where $\setK$ is a cone.

\begin{theorem}[Sufficient condition for quadratic cost]
\label{thm:quad}
Let $X$ be a face and $A'$ be the affine space $a_0+X^\perp$ for an arbitrary $a_0\in\aff(X)$.
Acute angles hold for the face $X$ and the quadratic cost
if and only if the projection of $\Omega$ (or, equivalently, $\M$) on $A'$ is contained in an affine acute
cone with the vertex $a_0$.
\end{theorem}

\begin{corollary}
\label{cor:quad:cube}
Acute angles hold for the quadratic cost and a hypercube $\Omega=\set{0,1}^n$.
\end{corollary}

\begin{corollary}
\label{cor:quad:simplex}
Acute angles hold for the quadratic cost and simplex
$\Omega=\set{e_i:\:i\in[n]}$ where $[n]=\set{1,2,\dotsc,n}$
and $e_i$ is the $i$-th vector of the standard basis in $\Re^n$.
\end{corollary}

\begin{theorem}[Log-partition over affinely independent outcomes]
\label{thm:log:partition}
If the set $\Omega$ is affinely independent then acute angles assumption
is satisfied for the log-partition cost.
\end{theorem}

\begin{theorem}[One-dimensional outcome spaces]
\label{thm:1d}
Acute angles hold for any cost function if $\M$ is a line segment.
\end{theorem}

Let $\Omega_1\subseteq\Re^{n_1}$ and $\Omega_2\subseteq\Re^{n_2}$ be outcome spaces with costs $C_1$ and $C_2$. We define the \emph{direct sum} of $\Omega_1$ and $\Omega_2$ to be
the outcome space $\Omega=\Omega_1\times\Omega_2$ with the cost $C:\Re^{n_1+n_2}\to\Re$
defined as
  $C(q_1,q_2)=C_1(q_1)+C_2(q_2)$.

\begin{theorem}[Acute angles for direct sums]
\label{thm:direct:sum}
If acute angles hold for $\Omega_1$ with cost $C_1$, and $\Omega_2$ with
cost $C_2$, then they also hold for their direct sum.
\end{theorem}

As a direct consequence of this theorem, we obtain that the log-partition
cost function satisfies the acute angles assumption on a hypercube. More generally,
any direct sum of costs on line segments satisfies the acute angles. This means that all
cost-based prediction markets consisting of independent binary questions are budget additive, regardless of
costs used to price individual questions.

As mentioned in the introduction, a vast number of deployed cost-based prediction
markets consists of independent questions (not necessarily binary), each priced according
to an LMSR (i.e., a log-partition cost on a simplex).
Theorems~\ref{thm:log:partition} and~\ref{thm:direct:sum} imply that this industry standard
is budget additive.


\subsection{Proof of \Thm{main}}
\label{sec:proof}
In this section we sketch the proof of \Thm{main} (for
a complete proof see \App{proof} of the full version).
We proceed in several steps. Let $\nu_0=p(q_0)$.
Assuming acute angles, we begin by constructing an oriented curve $L$
joining $\nu_0$ with $\mu$, by sequentially choosing portions of perpendiculars
for monotonically decreasing active sets. We then show
that budget additivity holds for any solutions with prices in $L$,
and finally show that the curve $L$
is the locus of the optimal prices of solutions $\hQ(q_0)$, as
well as optimal prices of solutions $\hQ(q)$ for any $q\in\hQ(q_0)$.

\textbf{Part 1: Construction of the solution path $L$.}
In this part, we construct:
\begin{itemize}[topsep=0pt]
\item a sequence of prices $\nu_0,\nu_1,\dotsc,\nu_k$ with $\nu_0=p(q_0)$ and $\nu_k=\mu$
\item a sequence of oriented curves $\ell_0,\dotsc,\ell_{k-1}$ where each $\ell_i$ goes from $\nu_i$ to $\nu_{i+1}$
\item a monotone sequence of sets $\Omega\supseteq X_0\supset X_1\supset\dotsb\supset X_k=\emptyset$, such that
      the following \emph{minimality} property holds: $X_i$ is the minimal face for
      all $\nu\in(\im\ell_i)\wo\set{\nu_{i+1}}$ for $i\le k-1$,
      and $X_k$ is the minimal face for $\nu_k$.
\item a sequence of states $q_1,\dotsc,q_{k-1}$ such that $q_i\in\pinv(\nu_i)\cap[q_{i-1}+\witness(X_{i-1})]$
\end{itemize}
The curves $\ell_i$ will be referred to as \emph{segments}.
The curve obtained by concatenating the segments $\ell_0$ through $\ell_{k-1}$ will
be called the \emph{solution path} and denoted $L$. In the special case
that $\nu_0=\mu$, we have $k=0$, $X_0=\emptyset$ and $L$ is a degenerate curve
with $\im L=\set{\mu}$.

If $\nu_0\ne\mu$, we construct the sequence of segments iteratively.
Let $X_0\ne\emptyset$ be the minimal face such that $\nu_0\in \conv(X_0\cup\set{\mu})$.
By the minimality, $\mu\not\in\aff(X_0)$.
Let $\gamma$ be the $\mu$-perpendicular to $\aff(X_0)$ at $q_0$. The curve $\gamma$ passes through $\nu_0$ and eventually reaches the boundary of $\conv(X_0\cup\set{\mu})$ at some $\nu_1$ by continuity of $\gamma$ (see \Thm{perp:cont}). Let segment $\ell_0$ be the portion of $\gamma$ going from $\nu_0$ to $\nu_1$.

This construction gives us the first segment $\ell_0$. There are two possibilities:
\begin{enumerate}[topsep=0pt]
\item $\nu_1=\mu$; in this case we are done;
\item $\nu_1$ lies on a lower-dimensional face of $\conv(X_0\cup\set{\mu})$; in this case, we pick
    some $q_1\in\pinv(\nu_1)\cap [q_0+\witness(X_0)]$, which can be done by the acute angles assumption, and use the above construction again, starting with $q_1$, and obtaining a new set $X_1 \subset X_0$ and a new segment $\ell_1$; and iterate.
\end{enumerate}
The above process eventually ends, because with each iteration, the size of the active set decreases.
This construction yields monotonicity of $X_i$ and the minimality property.

The above construction yields a specific sequence of $q_i\in\pinv(\nu_i)\cap[q_{i-1}+\witness(X_{i-1})]$. We show in \App{proof} of the full version that actually $q_i\in\pinv(\nu_i)\cap(q_0+X_{i-1}^\perp)$ and that the construction of $L$ is
independent of the choice of $q_1,q_2,\dotsc,q_{k-1}$.



\textbf{Part 2: Budget additivity for points on $L$.}
Let $\nu,\nu'\in\im L$ such that $\nu\preceq\nu'$. Let $q\in\hQ(\nu;q_0)$
and $q'\in\hQ(\nu';q)$ such that $q\in\hQ(B;q_0)$ and $q'\in\hQ(B';q)$. In this part we show that $q'\in\hQ(B+B';q_0)$.

First, consider the case that $\nu'=\mu$. To see that $q'\in\hQ(B+B';q_0)$, first note that the constraints of
\CP\ hold, because
$
    U(q',\omega;q_0)=U(q',\omega;q)+U(q,\omega;q_0)\ge -B'-B
$
for all $\omega$
by path independence of the utility function. As noted in the introduction, in the absence
of constraints, the utility $U(\bq,\mu;q_0)$ is maximized at any $\bq$ with $p(\bq)=\mu$. Thus,
$q'$ is a global maximizer of the utility and satisfies the constraints, so $q'\in\hQ(B+B';q_0)$.
If $\nu=\mu$, we must also have $\nu'=\mu$ and the statement holds by previous reasoning.

In the remainder, we only analyze the case $\nu\preceq\nu'\prec\mu$.
This means that $\nu\in(\im\ell_i)\wo\set{\nu_{i+1}}$ and $\nu'\in(\im\ell_j)\wo\set{\nu_{j+1}}$ for $i\le j$. By \Thm{KKT:final}, we therefore
must have $q\in[q_0+\witness(X_i)]$ and $q'\in[q+\witness(X_j)]$. By anti-monotonicity
of witness cones, $\witness(X_j)\supseteq\witness(X_i)$ and hence, $q'\in[q_0+\witness(X_j)]$, yielding
$
   q'\in\hQ(\nu';q_0).
$

We now argue that the budgets add up.
Let $x\in X_j\subseteq X_i$. By \Lem{KKT:initial}, we obtain that
$q\in\hQ(B;q_0)$ for
   $B=-U(q,x;q_0)$, and
$q'\in\hQ(B';q)$ for
   $B'=-U(q',x;q)$, and finally
$q'\in\hQ(\bB;q_0)$ for
  $ \bB=-U(q',x;q_0)$.
%
%
However, by path independence of the utility function
\[
   \bB=-U(q',x;q_0)
      =-U(q',x;q)-U(q,x;q_0)
      =B'+B.
\]

\textbf{Part 3: $L$ as the locus of all solutions.}
See \App{proof} of the full version for the proof that
\[
   \textstyle \hQ(q_0) = \bigcup_{\nu\in\im L} \hQ(\nu;q_0)
\enspace.
\]

\textbf{Part 3': $L$ as the locus of solutions starting at a midpoint.}
Let $\nu\in\im L$ and $q\in\hQ(\nu;q_0)$. Since $\hQ(\nu;q_0)\subseteq\pinv(\nu)\cap(q_0+X_\nu^\perp)$,
Part 1' (\App{proof} of the full version) yields that the solution path $L'$ for $q$ coincides with the portion of
$L$ starting at $\nu$. Applying the proof of Part 3 to $L'$, we obtain
\[
  \textstyle  \hQ(q) = \bigcup_{\nu'\in\im L: \nu'\succeq\nu} \hQ(\nu';q)
\enspace.
\]

\textbf{Part 4: Proof of the theorem.}
Let $B,B'\ge0$ and $q\in\hQ(B; q_0)$ and $q'\in\hQ(B';q)$. From Parts 3 and 3', we know that $q\in\hQ(\nu;q_0)$ and $q'\in\hQ(\nu';q)$ for some $\nu,\nu'\in\im L$ such that $\nu\preceq\nu'$.
By Part 2, we therefore obtain that $q'\in\hQ(B+B';q_0)$, proving the theorem.

\newpage
\bibliographystyle{plainnat}
\bibliography{prediction-markets}

\newpage
\appendix
\section{EXAMPLES OF COST FUNCTIONS}
\label{app:scoring:rules}

\begin{example}[Quadratic cost]
\label{ex:quad}
The first example of a cost function, applicable to arbitrary outcome sets $\Omega$, is the quadratic cost
function defined by $C(q)=\frac12\norm{q}^2$. In this case, $p(q)=q$, and
$U(q',\mu;q) = \frac12\norm{q-\mu}^2 - \frac12\norm{q'-\mu}^2$. It is clear that the expected utility is maximized
when $p(q')=q'=\mu$.

\emph{Convex conjugate and Bregman divergence:}
$C^*(\nu)=\frac12\norm{\nu}^2$ and $D(q,\nu)=\frac12\norm{q-\nu}^2$,
i.e., the Bregman divergence is a monotone transformation of the Euclidean distance.
\end{example}

\begin{example}[LMSR]
\label{ex:lmsr}
Our second example is Hanson's logarithmic market-scoring rule (LMSR)~\citep{Hanson03,Hanson07}, which is applied to
\emph{complete markets} whose outcomes coincide with basis vectors, i.e.,
$\Omega=\set{e_i:\:i\in [n]}$ where $e_i$ denotes the $i$-th basis vector and $[n]$
denotes the set $\set{1,\dotsc,n}$.
In this case $\M$ is the simplex in $\Re^n$ and beliefs $\mu$ are in one-to-one correspondence
with probability distributions over $\Omega$. The LMSR cost function is
\[
\textstyle
  C(q) = \ln\Parens{\sum_{i=1}^n e^{q[i]}}
\]
where $q[i]$ denotes the $i$-th coordinate of $q$. The price vector is
\[
  p(q)[i] = \frac{\partial C(q)}{\partial q[i]}
          =  \frac{e^{q[i]}}{\sum_{j=1}^n e^{q[j]}}
              = e^{q[i]-C(q)}
\enspace.
\]
For $\mu\in\M$, the expected utility function takes form
\begin{align*}
  U(q',\mu;q) & = \sum_{i=1}^n \mu[i]
  \BigParens{\ln p(q')[i] - \ln p(q)[i]} \\
             & = \KL(\mu\Vert p(q)) - \KL(\mu\Vert p(q'))
\enspace,
\end{align*}
where $\KL(\mu\Vert\nu)=\sum_{i=1}^n \mu[i]\ln(\mu[i]/\nu[i])$ is the KL-divergence.
KL-divergence is not symmetric, but it is non-negative, and zero only if the arguments are equal.
Thus, the expected utility is clearly maximized if and only if $\mu=p(q')$.

\emph{Convex conjugate and Bregman divergence:}
$C^*(\nu)=\infty$ if $\nu$ is not a probability measure on $\Omega$, and
$C^*(\nu)=\sum_{i=1}^n \nu[i]\ln\nu[i]$ otherwise, with the convention $0\ln 0=0$.
The Bregman divergence is
$D(q,\nu)=\KL(\nu\Vert p(q))$.
\end{example}

\begin{example}[Log-partition cost]
\label{ex:log-partition}
Next example is the log-partition function, which is applicable to arbitrary
outcome sets $\Omega$ and which
generalizes LMSR:
\[
\textstyle
  C(q) = \ln\Parens{\sum_{\omega\in\Omega} e^{q\inprod\omega}}
\enspace.
\]
Let $P_q$ be the probability measure over $\Omega$ defined by
\[
  P_q(\omega)=e^{q\inprod\omega-C(q)}
\enspace.
\]
The prices then correspond
to expected values of $\omega$ under $P_q$:
\[
  p(q) = \sum_{\omega\in\Omega} P_q(\omega)\omega
\enspace.
\]
For $\mu\in\M$, let $P_\mu$ denote the distribution of maximum entropy among $P$ with $\E_P[\omega]=\mu$ (this distribution is unique and always exists). Note that we are overloading notation on $P_q$ and $P_\nu$ and use the ``type'' of the subscript to
indicate which probability distribution we have in mind. The
expected utility function can be written as
\begin{align*}
  U(q',\mu;q) &
  = (q'-q)\inprod\E_{\omega\sim P_\mu}[\omega]-C(q')+C(q)
\\&{}
  = \E_{\omega\sim P_\mu}\Bracks{\ln P_{q'}(\omega) - \ln P_{q}(\omega)} \\
  &{}= \E_{\omega\sim P_\mu}\Bracks{\ln\Parens{\frac{P_\mu(\omega)}{P_{q}(\omega)}}
                             - \ln\Parens{\frac{P_\mu(\omega)}{P_{q'}(\omega)}}}
\\&{}
  = \KL(P_\mu\Vert P_q) - \KL(P_\mu\Vert P_{q'})
\enspace.
\end{align*}
A standard duality result shows that the infimum of $\KL(P_\mu\Vert P_{q'})$ over
the set $\set{P_{q'}:q'\in\Re^n}$ is zero. If there exists $q'$ attaining
this minimum, we must have $P_\mu=P_{q'}$ and thus $\mu=p(q')$. We argue
that the converse is true as well. Let $q',q''$ be such that $P_{q'}=P_\mu$ and $p(q')=p(q'')=\mu$. Then
by convexity of $C$, we have $C(q'')-C(q')=(q''-q')\cdot p(q')$. Therefore,
\[
   \KL(P_{q'}\Vert P_{q''})=(q'-q'')\inprod p(q') - C(q') + C(q'')=0
\enspace,
\]
i.e., $P_{q'}=P_{q''}=P_\mu$. Hence, for any $q\in\Re^n$,
$P_q$ is exactly the distribution
of maximum entropy among those $P$ that satisfy $\E_P[\omega]=p(q)$. In other words,
$P_{p(q)}=P_q$.

\emph{Convex conjugate and Bregman divergence:}
$C^*(\nu)=\infty$ if there is no distribution $P$ on $\Omega$ such that $\E_P[\omega]=\nu$,
 and $C^*(\nu)=\sum_{\omega\in\Omega} P_\nu(\omega)\ln P_\nu(\omega)$ otherwise.
 The Bregman divergence
$D(q,\nu)=\KL(P_\nu\Vert P_q)=\KL(P_\nu\Vert P_{p(q)})$.
\end{example}

\section{PROOF OF THEOREM 2.2}
\label{app:budget:constraints}

\begin{proof}[Proof of \Thm{budget:constraints}]
Throughout this proof we use concepts of convex conjugacy and Bregman divergence
introduced in \Sec{convex}.
Let $B\ge 0$ and $\B \coloneqq \set{q: U(q,\omega;q_0) \ge -B\text{ for all }\omega\in\Omega}$ be the set of states satisfying the constraints of \CP.
Using the definition of utility function, we can rewrite \CP\ as
\begin{equation}
\label{eq:primal}
   \max_{q\in\Re^n}~
\Bracks{
    (q-q_0)\inprod\mu - C(q)+C(q_0)
    -\ind_{\B}(q)
}
\end{equation}
where $\ind_{\B}(\cdot)$ is the convex indicator function, equal to $0$ on the set
$\B$ and $\infty$ outside it. Since the cost function
$C$ is convex on $\Re^n$, and $\B$ is closed, convex
and non-empty,
Fenchel's Duality Theorem~\citep[][Theorem 31.1]{Rockafellar70} implies that the
supremum of the above objective equals the following minimum
\begin{equation}
\label{eq:dual}
\min_{\nu\in\Re^n}~
\Bracks{
   C^*(\nu)-q_0\inprod\mu+C(q_0)
   +\ind_{\B}^*(\mu-\nu)
}
\end{equation}
and this minimum is attained at some $\hnu\in\Re^n$.
Now, let $\hq\in\hQ(B;q_0)$ be a solution of \Eq{primal}.
By Fenchel's Duality, the gap between the objectives
of \Eq{dual} and \Eq{primal} at $\hnu$ and $\hq$ must be
zero:
\begin{align*}
0
&{}=
   C^*(\hnu)-q_0\inprod\mu+C(q_0)
   +\ind_{\B}^*(\mu-\hnu) \\
&{} \quad -(\hq-q_0)\inprod\mu + C(\hq)-C(q_0)
    +\ind_{\B}(\hq)
\\
&{}=
   C^*(\hnu)-\hq\inprod\hnu + C(\hq)
   +\ind_{\B}^*(\mu-\hnu)-\hq\inprod(\mu-\hnu) \\
&{} \quad +\ind_{\B}(\hq)
\\
&{}=
   D(\hq,\hnu) +\Bracks{\ind_{\B}^*(\mu-\hnu)
   -\BigParens{\hq\inprod(\mu-\hnu)-\ind_{\B}(\hq)}}
\enspace.
\end{align*}
The term in the brackets is non-negative from the definition of the convex conjugate.
Since $D(\hq,\hnu)$ is also non-negative, we obtain that it must be zero,
i.e., $p(\hq)=\hnu$. Since this reasoning is independent of the choice
$\hq\in\hQ(B;q_0)$, the theorem follows.
\end{proof}

\section{OPTIMALITY CONDITIONS AND THE MINIMAL FACE}
\label{app:proof:KKT}

This appendix discusses arbitrage-free initialization,
provides proofs of optimality conditions (\Lem{KKT:initial} and \Thm{KKT:final}),
and shows that minimal faces are well defined.

\subsection{Arbitrage-free initialization}
\label{app:arb:free}

Throughout the paper we assume that the initial state $q_0$ is \emph{arbitrage-free}
in the sense that a trader with no budget prefers to stay in $q_0$:
\begin{definition}
We say that the initial state $q_0$ is arbitrage-free with respect to $\mu\in\M$
if $q_0$ is an optimal state at budget zero, i.e., $q_0\in\hQ(0;q_0)$.
\end{definition}
This corresponds to the assumption that
a trader cannot extract a positive expected profit without risking
some capital. Below we show that the condition is easily ensured for both the log-partition
and quadratic cost.

The assumption of arbitrage-free initialization was added after the paper was
published in the Proceedings of the 31st Conference on Uncertainty in
Artificial Intelligence, 2015.
For the sake of consistency with the published
version, we have only made corrections in Appendix. Without arbitrage-free initialization,
\Lem{KKT:initial} and its corollaries fail to hold for $B=0$.
The main text requires the following four corrections:
\begin{description}
\item[\Lem{KKT:initial}:] Let $q_0$ be arbitrage-free. Then $q\in\hQ(B;q_0)$ if and only if \dots
\medskip
\item[\Thm{KKT:final}:] Let $q_0$ be arbitrage-free. Then
                        $q\in\hQ(q_0)$ if and only if \dots
\medskip
\item[\Cor{budget:nu}:] Let $q_0$ be arbitrage-free. Then
                        $\hQ(\nu;q_0) = \pinv(\nu)\cap[q_0+\witness(X_\nu)]$. In particular, \dots
\medskip
\item[\Def{budget:additivity}:] \dots\
                        if for all beliefs $\mu\in\M'$ and
                        all arbitrage-free initial states $q_0\in\pinv(\M')$ \dots
\end{description}

%
Below we show that regardless of $q_0$, any $q\in\hQ(q_0)$ is arbitrage-free. This
property is useful in proving our main result (\Thm{main}).
We also derive necessary and sufficient conditions for arbitrage-free initialization.
In particular, we show that any $q_0$ is arbitrage-free for the log-partition
cost, and any $q_0$ with $p(q_0)\in\M$ is arbitrage-free for the quadratic cost.
\begin{proposition}
\label{prop:init:hQ}
Any $q\in\hQ(q_0)$ is arbitrage-free for any $q_0\in\Re^n$.
\end{proposition}
\begin{proof}
Let $q\in\hQ(B;q_0)$ for some $B\ge 0$ and let $q'\in\hQ(0;q)$. We
will use path independence to show that also $q\in\hQ(0;q)$.
To begin, note that the budget constraints satisfied by $q$ and $q'$ are
\begin{align*}
   &
   U(q,\omega; q_0)\ge -B
   &&
   \text{for all $\omega\in\Omega$,}
\\
   &
   U(q',\omega; q)\ge 0
   &&
   \text{for all $\omega\in\Omega$.}
\end{align*}
By path independence, we therefore have, for all $\omega\in\Omega$,
\[
  U(q',\omega;q_0) = U(q',\omega;q) + U(q,\omega;q_0) \ge -B
\enspace,
\]
so $q'$ satisfies the budget constraints for $B$ at the initial state $q_0$.
Since $q$ is optimal for the budget $B$ and the initial state $q_0$, we must have
\[
  U(q',\mu; q_0) \le U(q,\mu;q_0)
\enspace.
\]
Path independence then gives
\[
  U(q',\mu; q) = U(q',\mu; q_0) - U(q,\mu;q_0) \le 0 = U(q,\mu; q)
\;.
\]
Since $q'\in\hQ(0;q)$ and $q$ is a feasible action for the budget zero
and the initial state $q$, we must have $U(q',\mu;q)=0=U(q,\mu;q)$ and thus $q\in\hQ(0;q)$.
\end{proof}
\begin{proposition}
\label{prop:init:necessary}
If $q_0$ is arbitrage-free then $p(q_0)\in\M$.
\end{proposition}
\begin{proof}
For a contradiction, assume $p(q_0)\not\in\M$. We will show that there exists
a state $q'$ such that $U(q',\nu;q_0)\ge\eps$ for all $\nu\in\M$ and
some $\eps>0$. This will imply that the budget constraints for $B=0$ are satisfied
for $q'$ and also that $U(q',\mu;q_0)\ge\eps$. This contradicts
the optimality of $q_0$ because $U(q_0,\mu;q_0)=0$.

To proceed, consider the optimization
\begin{align}
\notag
   &\adjustlimits\sup_{q\in\Re^n}\min_{\nu\in\M}~
   U(q,\nu;q_0)
\\
\label{eq:init:1}
   &\quad{}
   =
   \adjustlimits\min_{\nu\in\M}\sup_{q\in\Re^n}~
   U(q,\nu;q_0)
\\
\notag
   &\quad{}
   =
   \adjustlimits\min_{\nu\in\M}\sup_{q\in\Re^n}
   \Bracks{
   q\inprod\nu - q_0\inprod\nu - C(q) + C(q_0)
   }
\\
\label{eq:init:3}
   &\quad{}
   =
   \min_{\nu\in\M}
   \Bracks{
   C^*(\nu) - q_0\inprod\nu + C(q_0)
   }
\\
\label{eq:init:4}
   &\quad{}
   =
   \min_{\nu\in\M} D(q_0,\nu)
\\
\label{eq:init:5}
   &\quad{}
   =
   D(q_0,\hnu)
\end{align}
where we use Sion's minimax theorem in \Eq{init:1}, the
definition of conjugate in \Eq{init:3}, and
denote the minimizer of \Eq{init:4} as $\hnu$ in \Eq{init:5}.
By assumption, $p(q_0)\not\in\M$ and thus $p(q_0)\ne\hnu$. This
implies that $D(q_0,\hnu)>0$, so the value of the initial supremum is greater
than zero. Therefore, there must exist $q'$ such that
$\min_{\nu\in\M}~U(q',\nu;q_0)\eqqcolon\eps>0$, yielding the desired contradiction.
\end{proof}

While $p(q_0)\in\M$ is a necessary condition for arbitrage-free
initialization, it is not sufficient. For example,
consider $n=1$, $\M=[0,1]$, and the cost function $C(q)=\max\set{0,q^2/2}$.
Here, $q_0=-1$ is not arbitrage-free for any $\mu\in(0,1]$, even though
$p(q_0)=0\in\M$. We next present a technical lemma followed
by two sufficient conditions.
\begin{lemma}
\label{lem:budget:0}
Let $q_0$ be such that $p(q_0)\in\M$ and let $q\in\hQ(0;q_0)$. Then $p(q)=p(q_0)$,
and for any decomposition of $p(q_0)$ into a convex combination over $\omega$,
i.e., for any weights $c_\omega\ge 0$ such that $p(q_0)=\sum_{\omega\in\Omega}c_\omega\omega$
and $\sum_{\omega\in\Omega} c_\omega=1$, we have $c_\omega U(q,\omega;q_0)=0$ for all $\omega$.
\end{lemma}
\begin{proof}
Budget constraints on $q\in\hQ(0;q_0)$ imply that
\begin{equation}
\label{eq:init:6}
   U(q,\omega;q_0)\ge 0\text{ for all $\omega\in\Omega$.}
\end{equation}
For a given set of weights $c_\omega$ write
\begin{align}
\label{eq:init:7}
  0
& \le \sum_\omega c_\omega U(q,\omega;q_0)
   =U\bigParens{q,\,\textstyle\sum_\omega c_\omega\omega;\,q_0}
\\
\notag
&  =U(q,p(q_0);q_0)
\\
\notag
&  =q\inprod p(q_0) - q_0\inprod p(q_0)-C(q) + C(q_0)
\\
\notag
&  =q\inprod p(q_0) - C^*\bigParens{p(q_0)} - C(q)
\\
\notag
& = -D\bigParens{q,p(q_0)}
\enspace.
\end{align}
Since Bregman divergence is non-negative, we obtain that $D(q,p(q_0))=0$ and thus $p(q)=p(q_0)$.
From \Eq{init:7}, we then also obtain that each of the terms $c_\omega U(q,\omega;q_0)$ must equal zero,
because $c_\omega\ge 0$ and utilities are non-negative by \Eq{init:6}.
\end{proof}

\begin{proposition}
\label{prop:init:sufficient}
The state $q_0$ is arbitrage-free if either of the following conditions holds:
\begin{enumerate}[label=\textup{(\alph*)},topsep=0pt]
\item
\label{init:ri}
$p(q_0)\in\ri\M$;

\item
\label{init:strict}
$p(q_0)\in\M$ and $C$ is strictly convex.
\end{enumerate}
\end{proposition}
\begin{proof}
~

\textbf{Part~\ref{init:ri}.}
Let $q\in\hQ(0;q_0)$. It suffices to show that $U(q,\mu;q_0)\le U(q_0,\mu;q_0)=0$.
To begin, note that since $p(q_0)\in\ri\M$, we can write $p(q_0)$ as a positive convex combination
of $\omega\in\Omega$, i.e., $p(q_0)=\sum_{\omega\in\Omega} c_\omega\omega$
where $c_\omega>0$ and $\sum_{\omega\in\Omega} c_\omega=1$. By \Lem{budget:0},
we obtain that $c_\omega U(q,\omega;q_0) = 0$ for all $\omega$
and hence $U(q,\omega;q_0)=0$ for all $\omega$. Since $\mu$ is a convex
combination of $\omega$, the linearity of utility in the second argument
yields $U(q,\mu;q_0)=0$.

\textbf{Part~\ref{init:strict}.}
We again appeal to \Lem{budget:0}. Let $q\in\hQ(0;q_0)$. Then by \Lem{budget:0},
we have that $p(q)=p(q_0)$, and the strict convexity of $C$ yields $q=q_0$.
\end{proof}

Part~\ref{init:ri}
implies that every $q_0\in\Re^n$ is arbitrage-free for the log-partition cost.
Part~\ref{init:strict}
implies that every $q_0$ such that $p(q_0)\in\M$ is arbitrage-free for
the quadratic cost.

\subsection{Proofs of Lemma~\ref{lem:KKT:initial} and Theorem~\ref{thm:KKT:final}}

\begin{proof}[Proof of \Lem{KKT:initial}]
We prove the revised version of the lemma
with the additional assumption
that $q_0$ is arbitrage-free (see \App{arb:free}).
First consider $B=0$ and assume $q\in\hQ(B;q_0)$.
Since $q_0$ is arbitrage-free,
we have $p(q_0)\in\M$, so it can be written
as a convex combination of $\omega\in\Omega$, say $p(q_0)=\sum_\omega c_\omega \omega$.
By \Lem{budget:0} we therefore obtain that
each of the terms $c_\omega U(q,\omega;q_0)$ must equal zero.
Thus, whenever $c_\omega>0$, we necessarily have $U(q,\omega;q_0)=0$. Let $S\coloneqq\set{\omega\in\Omega:\: c_\omega>0}$.
Note that for any $x,x'\in S$ we have $(q-q_0)\inprod(x'-x)=U(q,x';q_0)-U(q,x;q_0)=0$. Since $S$ is non-empty,
we can pick an arbitrary $s\in S$ and define $X\coloneqq\set{\omega\in\Omega:\: (q-q_0)\inprod(\omega-s)=0}$.
We next argue that $X$ satisfies conditions \ref{kkt:tight}--\ref{kkt:budget} and then show that it is actually a face.
Conditions \ref{kkt:tight} and \ref{kkt:budget} follow from the definition of $X$ and the fact that $U(q,s;q_0)=0$.
Condition \ref{kkt:nontight} follows because $q$ satisfies budget constraints for $B=0$.
Finally, condition \ref{kkt:conv} follows because
$S\subseteq X$. To see that $X$ is a face, note that condition \ref{kkt:nontight} actually shows that
$X$ is exactly the set of minimizers of the linear function $(q-q_0)\inprod\omega$ over $\omega\in\Omega$.

Now, consider $B>0$.
We begin by forming a Lagrangian of \CP, with non-negative multipliers $\lambda=(\lambda_\omega)_{\omega\in\Omega}$:
\[
  L(q,\lambda) = U(q,\mu;q_0) + \sum_\omega \lambda_\omega\Parens{U(q,\omega;q_0) + B}
\enspace.
\]
Since the utilities are convex and finite over $q\in\Re^n$, and $q_0$ is feasible with all of the constraints
satisfied with strict inequalities, KKT conditions are
both necessary and sufficient for optimality~\citep[Corollary 28.3.1]{Rockafellar70}. KKT conditions state that $q$ and $\lambda$
solve the above problem if and only if the following hold:
\begin{itemize}[topsep=0pt]
\item \emph{primal feasibility}: $U(q,\omega;q_0)\ge -B$ for all $\omega\in\Omega$;
\item \emph{dual feasibility}: $\lambda\ge 0$;
\item \emph{first-order optimality}: $\nabla_1 L(q,\lambda)=0$;
\item \emph{complementary slackness}: $\lambda_\omega\Parens{U(q,\omega;q_0)+B}=0$; for all $\omega\in\Omega$.
\end{itemize}

We next show that KKT conditions imply \ref{kkt:tight}--\ref{kkt:budget}. Assume that
KKT conditions hold. Let $X$ be
the set of outcomes with tight constraints, i.e., $X=\set{x\in\Omega: U(q,x;q_0)=-B}$. For this
$X$, the conditions \ref{kkt:tight} and \ref{kkt:nontight} hold by primal feasibility and our
definition of $X$. Note that we have either $X=\emptyset$ or $X=\argmin_{x\in\Omega} (q-q_0)\inprod x$, i.e., $X$ is a face of $\M$. If $X\ne\emptyset$, then $\ref{kkt:budget}$ follows from our definition of $X$. If $X=\emptyset$, then
$\ref{kkt:budget}$ follows by primal feasibility.
We prove \ref{kkt:conv} by analyzing
first-order optimality. First note that:
\[
   \nabla_1 U(q,\nu;q_0) = \nu - \nabla C(q) = \nu - p(q)
\enspace.
\]
Thus, first-order optimality is equivalent to
\begin{gather*}
   \nabla_1 U(q,\mu;q_0) + \sum_\omega \lambda_\omega\nabla_1 U(q,\omega;q_0)=0
\\
   \mu - p(q) + \sum_\omega \lambda_\omega(\omega-p(q))=0
\\
   p(q) = \frac{\mu+\sum_\omega\lambda_\omega\omega}{1+\sum_\omega\lambda_\omega}
\enspace.
\end{gather*}
By complementary slackness, $\lambda_\omega=0$ for $\omega\in\Omega\wo X$, so this shows \ref{kkt:conv}.

For the converse, assume that \ref{kkt:tight}--\ref{kkt:budget} hold. In particular,
by \ref{kkt:conv}, let $p(q)=c_\mu\mu + \sum_{\omega\in X} c_\omega\omega$ where $c_\mu,c_\omega\ge 0$
and $c_\mu+\sum_{\omega\in X}c_\omega=1$. If $c_\mu>0$, then we obtain
that KKT conditions hold for the given $q$ and $\lambda_\omega=c_\omega/c_\mu$ for $\omega\in X$
and $\lambda_\omega=0$ for $\omega\in\Omega\wo X$. Since KKT conditions are sufficient for optimality
\citep[Theorem 28.3]{Rockafellar70}, we obtain $q\in\hQ(B;q_0)$.

If $c_\mu=0$, we actually have $p(q)\in\conv(X)$ and $X\ne\emptyset$. Thus, from \ref{kkt:tight}
and \ref{kkt:budget},
\begin{align}
\notag
   B
   &=
   \sum_{\omega\in X}
   c_\omega\Bracks{-U(q,\omega;q_0)}
   =
   -U\bigParens{q,\,
   \textstyle\sum_{\omega\in X}c_\omega\omega;\,
   q_0}
\\
\label{eq:KKT:1}
   &=-U\bigParens{q,p(q);q_0}
\\
   &=-q\inprod p(q) + q_0\inprod p(q) + C(q) - C(q_0)
\\
   &=-C^*\bigParens{p(q)} + q_0\inprod p(q) - C(q_0)
\\
\label{eq:KKT:2}
   &=-D\bigParens{q_0,p(q)}
\enspace.
\end{align}
%
By non-negativity of Bregman divergence, \Eq{KKT:2}
can only hold if
$B=0$ and $p(q_0)=p(q)$. Feasibility of $q$ follows by
\ref{kkt:budget} and \ref{kkt:nontight}. To see that $q$ is also
optimal, first note that $U(q,\omega;q_0)\ge 0$ for all $\omega$.
Since $\mu$ is a convex combination of $\omega$, the linearity of
$U$ in the second argument implies $U(q,\mu;q_0)\ge 0$. However,
by arbitrage-free initialization, $q_0$ is optimal and $U(q_0,\mu;q_0)=0$,
so we must have $U(q,\mu;q_0)=0$ and $q$ optimal as well.
\end{proof}

\begin{proof}[Proof of \Thm{KKT:final}]
We prove the revised version of the theorem
with the additional assumption
that $q_0$ is arbitrage-free (see \App{arb:free}).
If $q\in[q_0+\witness(X_{p(q)})$ then we have $q\in\hQ(q_0)$ by \Lem{KKT:initial} with $X=X_{p(q)}$.
For the converse, assume that $q\in\hQ(q_0)$. \Lem{KKT:initial} then implies that there exists
a face $X$ such that $q\in[q_0+\witness(X)]$ and $p(q)\in\conv(X\cup\set{\mu})$. By minimality of $X_{p(q)}$,
we must have $X_{p(q)}\subseteq X$. By anti-monotonicity of witness cones, we then have $q\in[q_0+\witness(X_{p(q)})]$,
finishing the proof.
\end{proof}

\subsection{Minimal face}

\begin{proposition}
\label{prop:minimal}
Fix $\mu\in\M$. Then for any $\nu\in\M$, there exists the minimal face $X_\nu$ with
the following property: for any face $X$ such that $\nu\in\conv(X\cup\set{\mu})$,
we must have $X_\nu\subseteq X$.
\end{proposition}
\begin{proof}
If $\nu=\mu$ then $X_\nu=\emptyset$ and the statement holds. Otherwise, consider
the ray $\rho$ from $\mu$ towards $\nu$, and let $\nu'$ be the last point on the ray
that is contained in $\M$. Let $X_\nu$ be the unique face such that $\nu'$ lies
in the relative interior of $\conv(X_\nu)$.\footnote{The existence of such a unique face
follows by the standard result stating that relative interiors of $\conv(X)$
across non-empty faces $X$ form a disjoint partition of $\M$.}
We will argue that this face satisfies the condition stated in the proposition.
Let $X$ be any face such that $\nu\in\conv(X\cup\set{\mu})$. Then $\nu=\lambda\mu+(1-\lambda)\nu_X$
for $\nu_X\in\conv(X)$ and $\lambda\in[0,1)$. Since $\nu_X\in\M$, it must lie on the
ray $\rho$ at some point between $\nu$ and $\nu'$. We next argue than $\nu'\in\conv(X)$.
Suppose not, this means that $\nu'\ne\nu_X$, and $\nu_X$ maximizes some linear function, say $u\cdot\bnu$, over $\bnu\in\M$, and $u\cdot\nu'<u\cdot\nu_X$, i.e.,
\[
   u\cdot(\nu_X-\nu')>0
\enspace.
\]
Since $\nu\ne\mu$ and $\nu'\ne\nu_X$, and the points $\mu,\nu,\nu_X,\nu'$ lie on the ray $\rho$ (in that order), there exists $\eta>0$ such that $\mu-\nu_X=\eta(\nu_X-\nu')$ and thus
\[
  u\cdot(\mu-\nu_X)=\eta u\cdot(\nu_X-\nu') > 0
\]
implying that $u\cdot\mu>u\cdot\nu_X$ and contradicting the assumption that $\nu_X$ is the
maximizer. Thus, $\nu'\in\conv(X)$. By a similar reasoning, we can also show that for any
$x\in X_\nu$, we must have $x\in X$. Again, for the sake of contradiction assume that there
is $u$ such that $\nu'$ is a maximizer of $u\cdot\bnu$ over $\bnu\in\M$, but $x$ is not.
Then $x\ne\nu'$, and since $\nu'\in\ri\conv(X_\nu)$, for sufficiently small $\eta$, we
have $\nu''\coloneqq\nu'+\eta(\nu'-x)\in\M$, and $u\cdot\nu''>u\cdot\nu'$ contradicting
the maximizer property of $\nu'$. Thus, $X_\nu\subseteq X$.
\end{proof}

\section{IMPOSSIBILITY RESULT OF FORTNOW AND SAMI}
\label{app:FortnowS12}

We can use the KKT lemma (\Lem{KKT:initial}) and the continuity of the perpendiculars (\Thm{perp:cont})
to derive the impossibility result of \citet{FortnowS12}.
The result states that in the presence of budget constraints, there is
no market scoring rule guaranteeing that the market prices move towards the agent
belief along the connecting straight line (unless $\aff(\M)$ is a line or a point).

Our construction is based on the observation that according to the KKT lemma, the solutions
$q$ must lie on a perpendicular at $q_0$, and the continuity of the perpendiculars
implies that $p(q)$ is arbitrarily close to $p(q_0)$ for a small enough budget. In particular, if
$p(q_0)$ lies in the relative interior of $\conv(X_{p(q_0)}\cup\set{\mu})$ then we can assure that
so does $p(q)$, and this property does not change with small perturbations of $\mu$---in
particular, $q$ remains a solution to the budget constrained optimization.
Thus, the direction of movement of market prices is
independent of small changes in $\mu$.


This informal reasoning can be turned into the following formal argument.
Assume the dimension of $\aff(\M)$ is $d\ge 2$. Choose $\mu\in\ri\M$ and a face $X$
such that $\aff(X)$ is of dimension $d-1$, so $\conv(X\cup\set{\mu})$ is of dimension~$d$.
Pick $q_0$ such that $p(q_0)$ lies in $\ri\conv(X\cup\set{\mu})$, which assures
that $X$ is the minimal face for $p(q_0)$. Let $\nu_0$ denote $p(q_0)$.
Consider the $\mu$-perpendicular to $X$ at $q_0$. By continuity of perpendiculars,
we can pick a point $\nu\succ \nu_0$ on the perpendicular which is arbitrarily
close to $\nu_0$ and, in particular, which still lies in $\ri\conv(X\cup\set{\mu})$, so its minimal face is still $X$. Pick $q\in q_0+X^\perp$ such that $p(q)=\nu$, which
  is possible by \Prop{perp:dual}. We next show that actually $q\in q_0+\setK(X)$, which
implies that $q$ is a solution to budget constrained optimization
(by \Thm{KKT:final}).

Let $u$ be the normal to $\aff(X)$ within $\aff(\M)$ such that $\M$
lies in the non-negative half-space, i.e., all $\nu'\in\M$ can be expressed
in the form $\nu'=a'+t'u$ where $a'\in\aff(X)$ and $t'\ge 0$. Thus, we can write
\begin{align*}
   \nu_0&=a_0+t_0 u
   \enspace,
\\
   \nu&=a_\nu+t_\nu u
   \enspace,
\\
   \omega&=a_\omega + t_\omega u\quad\text{for all }\omega\in\Omega
   \enspace,
\end{align*}
for some $a_0,a_\nu,a_\omega\in\aff(X)$ and $t_0,t_\nu,t_\omega\ge0$.
By convexity
\begin{align}
\notag
  0&\le(q-q_0)\inprod(\nu-\nu_0)
\\
\notag
   &=(q-q_0)\inprod\Bracks{a_\nu-a_0+(t_\nu-t_0)u}
\\
\label{eq:FS:1}
   &=(q-q_0)\inprod(t_\nu-t_0)u
\end{align}
where \Eq{FS:1} follows because $(q-q_0)\perp\aff(X)$. Since $\nu\succ\nu_0$ along the perpendicular,
we have that $t_\nu>t_0$ and so \Eq{FS:1} implies
\[
  0\le (q-q_0)\inprod u
\enspace.
\]
Thus, for any $\omega\in\Omega$ and $x\in X$,
\begin{align*}
  (q-q_0)\inprod(\omega-x)
   &=(q-q_0)\inprod(a_\omega-x+t_\omega u)
\\
   &=t_\omega(q-q_0)\inprod u
\\
   &\ge 0
\enspace,
\end{align*}
showing that $q\in q_0+\setK(X)$, i.e., $q\in\hQ(q_0)$. If $\nu_0,\nu$ and $\mu$ are not
on a straight line, we are done. Otherwise, slightly move $\mu$ within
the affine space parallel to $\aff(X)$, so that $X$ remains the minimal face for $\nu$
and thus $q$ remains a solution, but $\nu_0,\nu$ and $\mu$ are no longer on
a straight line.

\section{BUDGET ADDITIVITY: EXAMPLES}
\label{app:additivity:examples}

Using the KKT lemma, we illustrate on examples that budget additivity sometimes
 holds and sometimes does not. Recall
that budget additivity states that if several agents have the same belief and limited budgets,
the sequence of their actions is equivalent to the action of a single agent with the same belief
and the sum of the budgets. In the first example, we give an illustration of when this
property holds. In the second example, we show how this property can be violated,
and the single agent with the sum of budgets has more power in the market.

\begin{example}[Quadratic cost on a square]
Consider the following outcome space and belief:
\[
\renewcommand{\arraycolsep}{0pt}
\renewcommand{\arraystretch}{1.1}
\begin{array}{rl}
\omega_{00} &{}= (0,\,0)
\\
\omega_{01} &{}= (0,\,1)
\\
\omega_{10} &{}= (1,\,0)
\\
\omega_{11} &{}= (1,\,1)
\\
\mu &{}= (0.9,\,0.3)
\end{array}
\]
Further, consider the following market states:
\[
\renewcommand{\arraycolsep}{0pt}
\renewcommand{\arraystretch}{1.1}
\begin{array}{rl}
  q_0 &{}= \nu_0 = (0.5,\,0.1)
\\
  q_1 &{}= \nu_1 = (0.6,\,0.2) = \frac13\omega_{00}+\frac23\mu
\\
  q_\mu &{}= \mu
\end{array}
\]
The divergence of these states (and the belief $\mu$) from individual outcomes is:
\[
\begin{array}{c|cccc}
\frac12\norm{\cdot-\cdot}^2
& \omega_{00}
& \omega_{01}
& \omega_{10}
& \omega_{11}
\\
\hline
\nu_0
& 0.13 & 0.53 & 0.13 & 0.53
\\
\nu_1
& 0.2 & 0.5 & 0.1 & 0.4
\\
\mu
& 0.45 & 0.65 & 0.05 & 0.25
\end{array}
\]
With these in hand, we can now use the KKT lemma and show that $q_1=\nu_1$ is
an optimal action at $q_0=\nu_0$ under belief $\mu$ for a specific budget.
Since $q_1$ is a convex combination of $\omega_{00}$ and $\mu$, we need to show that the only tight budget constraint is due to $\omega_{00}$. We also calculate budgets required to move from $q_0$ and $q_1$ to $q_\mu$:
\[
\begin{array}{c|c@{~}c@{~}c@{~}c|c}
& \omega_{00}
& \omega_{01}
& \omega_{10}
& \omega_{11}
&
\\
\hline
U(q_1,\cdot;\,q_0)
& -0.07 & 0.03 & 0.03 & 0.13
& B_{01}=0.07
\\
U(q_\mu,\cdot;\,q_0)
& -0.32 & -0.12 & 0.08 & 0.28
& B_{0\mu}=0.32
\\
U(q_\mu,\cdot;\,q_1)
& -0.25 & -0.15 & 0.05 & 0.15
& B_{1\mu}=0.25
\end{array}
\]
Hence, a sequence of moves with budgets $B_{01}$ and $B_{1\mu}$ is
equivalent to a single move with the budget $B_{0\mu}=B_{01}+B_{1\mu}$.
While we have shown this only for a specific sequence of budgets,
results of \Sec{main} show that for the quadratic cost on a square, budget
additivity holds for any sequence of budgets and any belief $\mu\in\M$.
\end{example}

\begin{example}[Quadratic cost on an obtuse triangle.]
\label{example:obtuse:long}
Now, we work out an example where the budget additivity
does not hold.
Consider the following outcome space and belief:
\[
\renewcommand{\arraycolsep}{0pt}
\renewcommand{\arraystretch}{1.1}
\begin{array}{rl}
  \omega_1 &{}= (0.0,\,0.0)
\\
  \omega_2 &{}= (1.8,\,0.0)
\\
  \omega_3 &{}= (6.0,\,4.2)
\\
  \mu &= (2.7,\,1.8)
\end{array}
\]
Further, consider the following set of market states:
\[
\renewcommand{\arraycolsep}{0pt}
\renewcommand{\arraystretch}{1.1}
\begin{array}{rl}
  q_0 &{}= \nu_0 = (2.7,\,0.9)
\\
  q_1 &{}= \nu_1 = (2.4,\,1.2) = \frac13\omega_2 + \frac23\mu
\\
  q_2 &{}= \nu_2 = (2.4,\,1.6) = \frac19\omega_1 + \frac89\mu
\\
  q_3 &{}= \nu_3 = \Parens{0.9\sqrt{\frac{105}{13}},\,
                   0.6\sqrt{\frac{105}{13}}}
\\
      &{}= \Parens{1-\frac13\sqrt{\frac{105}{13}}}\omega_1
           + \Parens{\frac13\sqrt{\frac{105}{13}}}\mu
         \approx \frac{1}{19}\omega_1 + \frac{18}{19}\mu
\\
  q_\mu &{}= \mu
\end{array}
\]
The divergence of these states (and the belief $\mu$) from individual outcomes is:
\[
\begin{array}{c|ccc}
\frac12\norm{\cdot-\cdot}^2
& \omega_1
& \omega_2
& \omega_3
\\
\hline
\nu_0
& 4.05 & 0.81 & 10.89
\\
\nu_1
& 3.6 & 0.9 & 10.98
\\
\nu_2
& 4.16 & 1.46 & 9.86
\\
\nu_3
& 4.725 & 1.74\dots & 9.04\dots
\\
\mu
& 5.265 & 2.025 & 8.325
\end{array}
\]
Again as before, we can  use the KKT lemma and show for $j=1,2,3$,
that $q_j=\nu_j$ is an optimal action at $q_{j-1}=\nu_{j-1}$ under belief $\mu$, with the
corresponding budgets as:
\[
\renewcommand{\arraycolsep}{4pt}
\begin{array}{c|c@{~}c@{~}c|c}
& \omega_1
& \omega_2
& \omega_3
&
\\
\hline
U(q_1,\cdot;\,q_0)
& 0.45 & -0.09 & -0.09
& B_{01}=0.09
\\
U(q_2,\cdot;\,q_1)
& -0.56 & -0.56 & 1.12
& B_{12}=0.56
\\
U(q_3,\cdot;\,q_2)
& -0.565 & -0.28\dots & 0.82\dots
& B_{23}=0.565
\\
U(q_\mu,\cdot;\,q_0)
& -1.215 & -1.215 & 2.565
& B_{0\mu}=1.215
\end{array}
\]
The above table also shows that the budget $B_{0\mu}=1.215$ suffices to move directly
from $q_0$ to $q_\mu$. However, note that the sum
\[
 B_{01}+B_{12}+B_{23}=1.215=B_{0\mu}
\enspace,
\]
but $\nu_3\ne\mu$, i.e., after the sequence of optimal actions with budgets $B_{01}$,
$B_{12}$, and $B_{23}$, the market is still not at the belief shared by all agents,
even though with the budget $B_{0\mu}$, it would have reached it. Note that
it is possible to achieve budget additivity by using log-partition cost instead
of quadratic cost (\Thm{log:partition}).
\end{example}

\section{PERPENDICULARS}
\label{app:perp}

\subsection{Proofs of Propositions~\ref{prop:perp:dual} and~\ref{prop:perp:equiv}}

%
%
\begin{proof}[Proof of \Prop{perp:dual}]
We will show that condition (i) is equivalent to condition (ii) by analyzing
the first order optimality conditions. Consider the problem
\begin{equation}
\label{problem:perp}
  \min_{\nu'\in A_\lambda}~
  D(q,\nu')
\end{equation}
used to define $\nu_\lambda$. Assume that the minimum is attained
at some $\nu'\in\ri\dom C^*$. Thus, $\nu'\in A\cap(\ri\dom C^*)$.
Since $D(q,\nu')$ is subdifferentiable at $\nu'$,
the first order optimality implies that
\begin{equation}
\label{cond:perp:opt}
   \bigParens{\partial_2 D(q,\nu')}\cap A_\lambda^\perp\ne\emptyset
\enspace.
\end{equation}
Since $\partial_2 D(q,\nu')=\partial C^*(\nu')-q=\pinv(\nu')-q$,
and $A_\lambda^\perp=A_0^\perp$,
we have
\[
   \pinv(\nu')\cap(q+A_0^\perp)\ne\emptyset
\enspace,
\]
proving that (i)$\Rightarrow$(ii).
Conversely, assume that $\nu'\in A\cap(\ri\dom C^*)$ and $\pinv(\nu')\cap(q+A_0^\perp)\ne\emptyset$.
Then we can pick $\lambda$ such that $\nu'\in A_\lambda$, and for this $\lambda$,
we obtain that condition \eqref{cond:perp:opt} holds and
hence $\nu'$ solves problem \eqref{problem:perp}. Since $\nu'\in\ri\dom C^*$,
we obtain that $\nu'\in\im\gamma$.
\end{proof}

%
\begin{proof}[Proof of \Prop{perp:equiv}]
Let $\gamma'$ be the $a_1$-perpendicular to $A_0$ at $q'$. Since the ambient
space $A$ for both perpendiculars is the same, by \Prop{perp:dual}(ii), it suffices
to show that $q+A_0^\perp=q'+A_0^\perp$. However, this follows by the assumption
of the theorem, since $q'-q\in A_0^\perp$.
\end{proof}

\subsection{Continuity of perpendiculars}
\label{app:perp:cont}

In this section, we prove two important properties of
perpendiculars:
(a) they are continuous maps;
(b) intersections of perpendiculars with compact convex sets
correspond to compact sets of market states up to certain ``irrelevant displacements''.
To define these irrelevant displacements,
let $\setL$ be the linear space parallel to $\aff(\dom C^*)$. Then the displacements of
market state within $\setL^\perp$ are \emph{irrelevant} in the sense that
they have no effect on the Bregman divergence and
hence by \Eq{U} also no effect on the utility function. Specifically, $D(q+u,\nu)=D(q,\nu)$
for all $u\in\setL^\perp$ (see next proposition). For instance, for LMSR over a simplex,
the irrelevant displacements are of the form $\lambda\vec{1}$ where $\lambda\in\Re$ and $\vec{1}$ is the all-ones vector.

\begin{proposition}
Let $\setL$ be the linear space parallel to $\aff(\dom C^*)$.
Then for all $q\in\Re^n$ and $u\in\setL^\perp$
\[
  D(q+u,\nu)=D(q,\nu)
\enspace.
\]
\end{proposition}
\begin{proof}
If $\nu\not\in\dom C^*$ then the statement obviously holds.
Pick $\nu\in\dom C^*$, $q\in\Re^n$ and $u\in\setL^\perp$. By the Mean Value
Theorem, we can write
\[
  C(q+u)-C(q)=u\inprod\nabla C(\bq)
\]
for some $\bq$. Let
$\bnu\coloneqq\nabla C(\bq)\in\dom C^*$. Then we can write
\begin{multline*}
  D(q+u,\nu)-D(q,\nu)
\\
  =
  C(q+u)-C(q)-u\inprod\nu
  =u\inprod(\bnu-\nu) = 0
\end{multline*}
since $u\perp(\bnu-\nu)$.
\end{proof}

The following result of~\citet{Rockafellar70} will be instrumental in proving
continuity properties of the perpendicular. It is paraphrased for our setting. The notation $\interior$ refers to the topological interior of the set.

\begin{theorem}[Theorem 24.7 of \citet{Rockafellar70}]
\label{thm:compact}
Let $G:\Re^n\to(-\infty,\infty]$ be a lower semi-continuous convex function, and let
$K$ be a non-empty, closed and bounded subset of $\interior(\dom G)$. Then the set
\[
      \partial G(K) = \bigcup_{u\in K}\partial G(u)
\]
is non-empty, closed and bounded.
\end{theorem}

%
%
%
%
%
%
%

Now we are ready to state and prove the continuity of
perpendiculars:

\begin{theorem}
\label{thm:perp:cont}
Let $\gamma$ be the $a_1$-perpendicular to $A_0$ at $q$, and $K\subseteq\ri\dom C^*$ be a closed bounded convex set intersecting $\im\gamma$.
\begin{itemize}[topsep=0pt]
\item[\textup{(a)}]
  The map $\gamma$ is continuous and $\Lambda$ is open.
\item[\textup{(b)}]
  The intersection
  $M\coloneqq\set{(\nu,q):\:\nu\in(\im\gamma)\cap K,\,q\in\pinv(\nu)}$
  can be written as
  $M=\setC+(0\oplus\setL^\perp)$ where $\setC$ is compact and $\oplus$ denotes a direct sum of vector spaces.
\end{itemize}
\end{theorem}

\begin{proof}
Throughout the proof, let $F(\nu)\coloneqq D(q,\nu)=C(q)+C^*(\nu)-q\inprod\nu$. Note
that $F$ is strictly convex on $\ri\dom C^*$.
We will be also making frequent use of the fact that $F$ is continuous
on $\ri\dom C^*$ (because $C^*$ is continuous on $\ri\dom C^*$
by Theorem 10.1 of \citet{Rockafellar70}). Let $\norm{\cdot}$ denote the usual Euclidean norm.
Let $a_0=\argmin_{a\in A_0}\norm{a_1-a}$, i.e., $(a_1-a_0)\in A_0^\perp$.
Let $A=\aff(A_0\cup\set{a_1})$ and recall that
\[
  A_\lambda = A_0+\lambda(a_1-a_0)
\]
and
\[
  \nu_\lambda=\argmin_{A_\lambda} F(\nu)
\enspace.
\]
We use the notation $\setB(\nu,r;M)\coloneqq\set{\nu'\in M:\:\norm{\nu'-\nu}\le r}$
for the Euclidean ball relative to set $M$, and $\setS(\nu,r;M)\coloneqq\set{\nu'\in M:\:\norm{\nu'-\nu}=r}$
for the Euclidean sphere relative to set $M$.

\paragraph{Part (a).}
We need to show that $\gamma$ is continuous.
Let $\lambda\in\Lambda$, i.e., $\nu_\lambda\in\ri\dom C^*$.
Choose a sufficiently small $r>0$
such that the ball
$\setB\coloneqq\setB(\nu_\lambda,r;A)$ is contained in $\ri\dom C^*$. To show the continuity of $\gamma$
and openness of $\Lambda$, it suffices to show that if $\lambda'$ is close enough to $\lambda$ then $\nu_{\lambda'}\in\setB$.

Let $\eps=r/\sqrt{2}$. Consider
the sphere $\setS_\lambda\coloneqq\setS(\nu_\lambda,\eps;A_\lambda)\subseteq\ri\dom C^*$.
This sphere is a compact set, so $F$ attains the minimum on $\setS_\lambda$. By
strict convexity of $F$ and the optimality of $\nu_\lambda$, this minimum must be bounded
away from $F(\nu_\lambda)$. Thus, there exists $\delta>0$ such that
\begin{equation}
\label{eq:sep}
  F(\nu)\ge F(\nu_\lambda)+\delta
  \text{ for all }\nu\in\setS_\lambda
\enspace.
\end{equation}
Let $\delta'=\delta/3$.
Since $F$ is continuous on $\ri\dom C^*$, it is uniformly
continuous on $\setB$ and thus there exists $\eps'\in(0,\eps]$ such that
\begin{equation}
\label{eq:uniform}
\begin{aligned}
  &\abs{F(\nu')-F(\nu)}\le\delta'
\\&\quad\text{for all $\nu,\nu'\in\setB$ such that $\norm{\nu'-\nu}\le\eps'$.}
\end{aligned}
\end{equation}
Let $\setB_\lambda\coloneqq\setB(\nu_\lambda,\eps;A_\lambda)$
be the closed ball with $\setS_\lambda$ as the border.
For any $\lambda'$, let $\setS_{\lambda'}\coloneqq\setS_\lambda+(\lambda'-\lambda)(a_1-a_0)\subseteq A_{\lambda'}$ and similarly $\setB_{\lambda'}$. Let $\tnu_{\lambda'}=\nu_\lambda+(\lambda'-\lambda)(a_1-a_0)\in A_{\lambda'}$.

Note that if $\abs{\lambda'-\lambda}\le\eps'$, then
$\setB_{\lambda'}\subseteq\setB$,
because $\sqrt{(\lambda'-\lambda)^2+\eps^2}\le\eps\sqrt{2}=r$. So we can
use the above uniform continuity result and write:
\begin{itemize}
\item
$F(\nu')\ge F(\nu_\lambda)+\delta-\delta'=F(\nu_\lambda)+2\delta'$
for all $\nu'\in S_{\lambda'}$\\
\null\hfill
by Eqs.~\eqref{eq:sep} and~\eqref{eq:uniform}
\smallskip

\item
$F(\tnu_{\lambda'})\le F(\nu_\lambda)+\delta'$
\hfill
by \Eq{uniform}
\end{itemize}
By convexity of $F$, this means that $\nu_{\lambda'}\in\setB_{\lambda'}$. This proves
that $\nu_{\lambda'}\in\setB$ provided that $\abs{\lambda'-\lambda}\le\eps'$,
thus proving the continuity of $\gamma$ at $\lambda$.

\paragraph{Part (b).}
We first show that the set $M$ is closed and then that it is bounded, except
for directions in $0\oplus\setL^\perp$. Since $K\subseteq\ri\dom C^*$, we
can use \Prop{perp:dual} to write the set $M$ as
\begin{equation}
\label{eq:M}
\begin{array}{r@{~}l}
   M = & \set{(\nu',q'):\:\nu'\in\Re^n,\,q'\in\partial C^*(\nu')} \\
       & ~ \cap \Parens{\Re^n\times(q+A_0^\perp)} \cap (K\times\Re^n) ~,
\end{array}
\end{equation}
where we used the identity $\pinv(\nu')=\partial C^*(\nu')$ valid for all $\nu'$.
The closedness follows, because the set of pairs $\set{(\nu',q'):\:\nu'\in\Re^n,\,q'\in\partial C^*(\nu')}$ is closed
\citep[Theorem 24.4]{Rockafellar70}.

Denote the projections of $M$ on its two components as
\begin{align*}
  M_1 & \coloneqq\set{\nu':\:(\nu',q')\in M\text{ for some $q'$}}
\enspace, \\
  M_2 & \coloneqq\set{q':\:(\nu',q')\in M\text{ for some $\nu'$}}
\enspace.
\end{align*}
To show boundedness, we only need to analyze $M_2$ since $M_1\subseteq K$.
By \Eq{M}, it in fact suffices to show that the set $\partial C^*(K)=\bigcup_{\nu\in K}\partial C^*(\nu)$ is bounded except for directions in $\setL^\perp$.
We would like to appeal to \Thm{compact}, but we cannot do it directly,
because it is stated
for the \emph{interior} rather than the \emph{relative interior}.
For $\nu\in\ri\dom C^*$, we have
$\partial C^*(\nu)\ne\emptyset$, and using the fact that $C^*(\nu)=\infty$ over $\nu+(\setL^\perp\wo\set{0})$,
we obtain that
\begin{equation*}
   \partial C^*(\nu) = S + \setL^\perp
\end{equation*}
for some set $S\subseteq\setL$. This set $S$ coincides with subdifferential when $C^*(\nu)$ is
only viewed as a function over $\aff(\dom C^*)$. By applying \Thm{compact} to this restriction,
we then indeed obtain that
\[
   \partial C^*(K) = \setC + \setL^\perp
\]
for a non-empty closed and bounded set $\setC$. Note that $\setL^\perp\subseteq A_0^\perp$, so $\setL^\perp$ survives
taking the intersection in \Eq{M} and hence part (b) of the theorem follows.
\end{proof}

\section{PROOFS OF SUFFICIENT CONDITIONS FOR ACUTE ANGLES}
\label{app:acute:angles}

\begin{proposition}
\label{prop:Eacute}
For the quadratic cost, \Def{acute} is equivalent to \Def{Eacute}.
\end{proposition}
\begin{proof}
We first show that \Def{acute} (general acute angles) implies \Def{Eacute} (Euclidean acute angles).
Assume that the general acute angles
hold for $X$. Let $\bnu\in\M$ and $\nu$ be
its projection on $\aff(X)$. If $\bnu=\nu$ then angles between $\bnu$, $\nu$ and $\omega\in\Omega$
are non-obtuse in the sense that $(\bnu-\nu)\cdot(\omega-\nu)\ge 0$. If $\bnu\ne\nu$, then let
$\gamma$ be the $\bnu$-perpendicular to $X$ at $\nu$ (note that $p$ is the identity map, so $\nu$
is both a state and the corresponding price vector).
Note that $\bnu\succeq\nu$ and thus by the general acute angles assumption $\pinv(\bnu)\cap[\nu+\witness(X)]\ne\emptyset$. Since $\pinv(\bnu)=\set{\bnu}$, this is equivalent to
\[
   (\bnu-\nu)\cdot(\omega-x)\ge 0\text{ for all }x\in X,\omega\in\Omega
\enspace.
\]
Since $\nu\in\aff(X)$, we obtain that $(\bnu-\nu)\cdot(\omega-\nu)\ge 0$, i.e., the Euclidean acute
angles hold.

Conversely, assume that the Euclidean acute angles hold. Let $\gamma$ be a $\mu$-perpendicular to a face $X$
for some $\mu\in\M$ and $\nu'\succeq\nu$ be two points in $\im\gamma$ such that $\nu\in\M$. We need to show that $\pinv(\nu')\cap[\nu+\witness(X)]\ne\emptyset$, which is equivalent to
\begin{equation}
\label{eq:Gacute}
   (\nu'-\nu)\cdot(\omega-x)\ge 0\text{ for all }x\in X,\omega\in\Omega
\enspace.
\end{equation}
If $\nu'=\nu$ then \eqref{eq:Gacute} holds. Otherwise,
we can write $\nu'=\nu+u_X+\lambda(\mu-\hnu)$ for an arbitrary $\hnu\in\aff(X)$,
a suitable $\lambda>0$ and $u_X$ from the linear space parallel with $\aff(X)$. Pick $\hnu\in\ri\conv(X)$ (and the
corresponding $\lambda$ and $u_X$). We claim that there is a small enough $\eta>0$ such that
$\bnu\coloneqq\hnu+\eta(\nu'-\nu)\in\M$. This follows, because from our previous reasoning,
\[
  \bnu=\hnu+\eta u_X + \eta\lambda(\mu-\hnu),
\]
and for sufficiently small $\eta>0$, we have $[\hnu+\eta\lambda(\mu-\hnu)]\in\ri\conv(X\cup\set{\mu})$ and
then also for sufficiently small $\eta$, $[\hnu+\eta\lambda(\mu-\hnu)+\eta u_X]\in\ri\conv(X\cup\set{\mu})\subseteq\M$.
Thus, by the Euclidean acute angles,
\[
  (\bnu-\hnu)\cdot(\omega-\hnu)\ge 0\text{ for all }\omega\in\Omega.
\]
Since $\bnu-\hnu=\eta(\nu'-\nu)$ and $(\nu'-\nu)\perp(x-\hnu)$ for all $x\in X$, we also obtain
\[
  \eta(\nu'-\nu)\cdot(\omega-x)\ge 0\text{ for all }x\in X,\omega\in\Omega
\]
proving \eqref{eq:Gacute} and finishing the proof.
\end{proof}

\begin{proof}[Proof of \Thm{quad}]
Let $L$ be the linear space parallel to $\aff(X)$.
First show that the acute angles imply the inclusion of the projection
in an acute cone. Note that the inclusion is either true for all $a_0\in\aff(X)$
or none, so we can without loss of generality choose $a_0\in\conv(X)$.
Let $\omega_1,\omega_2\in\Omega$ and let $\omega'_1$ and $\omega'_2$
be their projections to $A'$, thus
\[
  \omega'_1-\omega_1\in L
\enspace,
\quad
  \omega'_2-\omega_2\in L
\enspace.
\]
We need to show that
\[
  (\omega'_1-a_0)\cdot(\omega'_2-a_0)\ge 0
\enspace.
\]
If $\omega_1\in\aff(X)$ then $\omega'_1=a_0$ and the statement holds.
Assume that $\omega_1\not\in\aff(X)$ and let $\gamma$ be the $\omega_1$-perpendicular
to $X$ at $a_0$. Let $\omega''_1\in\im\gamma$ be the projection of $\omega_1$ on $\im\gamma$.
Thus, we also have
\[
  \omega''_1-\omega_1\in L
\]
and also $\omega''_1\succeq a_0$. Now by the acute angles assumption,
$\omega''_1-a_0\in\witness(X)$, i.e., for any $x\in X$,
\[
  0\le(\omega''_1-a_0)\cdot(\omega_2-x)
\enspace.
\]
Combining this with the previous identities, we obtain
\begin{align*}
  0 & \le(\omega''_1-a_0)\cdot(\omega_2-x)
   =(\omega''_1-a_0)\cdot(\omega'_2-a_0) \\
   & =(\omega'_1-a_0)\cdot(\omega'_2-a_0)
\end{align*}
where the first equality follows because $\omega''_1-a_0\in L^\perp$ and
\[
  \omega'_2-\omega_2\in L
\enspace,
\quad
  a_0-x\in L
\enspace,
\]
the second equality follows because $\omega'_2-a_0\in L^\perp$ and
\[
\omega''_1-\omega'_1=[(\omega''_1-\omega_1)-(\omega'_1-\omega_1)]\in L
\enspace.
\]

For the converse, assume that the inclusion of the projection of $\M$
in an affine acute cone
holds. Let $\gamma$ be the $\mu$-perpendicular to $X$ at $\nu$ for some $\mu,\nu\in\M$
and let $\nu'\succeq\nu$. We need to show that $\nu'-\nu\in\witness(X)$. Note
that $0\in\witness(X)$, so we only analyze $\nu'\ne\nu$. Let $\mu'$ be the projection
of $\mu$ on $\im\gamma$ and $a_0$ be the intersection
 of $\im\gamma$ with $\aff(X)$. Note that $\nu'-\nu=\eta(\mu'-a_0)$ for
 a suitable $\eta>0$, so it suffices to show that $\mu'-a_0\in\witness(X)$.
 Pick $\omega\in\Omega$ and $x\in X$ and let
 $\omega'$ be the projection of $\omega$ into $A'\coloneqq a_0+X^\perp$. Since
 the projection of $\M$ into $A'$ is contained in an affine acute cone with the
 vertex $a_0$, we obtain
\[
 (\mu'-a_0)\cdot(\omega'-a_0)\ge 0
\]
Since $\omega'-\omega\in L$ and $x-a_0\in L$, whereas $\mu'-a_0\in L^\perp$, we obtain
\[
 (\mu'-a_0)\cdot(\omega-x)\ge 0
\]
showing that the acute angles hold.
\end{proof}

\begin{proof}[Proof of \Cor{quad:cube}]
We will show that the assumption of \Thm{quad} holds.
Since the assumption is invariant under rigid
transformations, we can just consider the case $a_0=0\in X$.
In this case, the projection of $\Omega$ is a lower dimensional hypercube
(corresponding to a subset of $\Omega$). Note that $\Omega$
lies in the non-negative orthant and the non-negative orthant is
an acute cone with the vertex $a_0=0$, so the assumption of
\Thm{quad} holds and hence the acute angles hold
for the hypercube.
\end{proof}

\begin{proof}[Proof of \Cor{quad:simplex}]
Again, by symmetry, it suffices to consider faces of
the form $X=\set{e_i:\:i\in[k]}$ for $k\in\set{1,\dotsc,n}$.
Let $a_0=e_1$. The affine space $A'$ is described by
\begin{align*}
  A' &{} = \set{a\in\Re^n:\: (a-e_1)\inprod(e_i-e_1)=0 \text{ for } 2\le i\le k} \\ 
  &{} = \set{a\in\Re^n:\: a[i] = a[1]-1 \text{ for } 2\le i\le k}
\end{align*}
where we use notation $a[i]$ to denote the $i$-th coordinate.
The projection of $e_j$ for $j>k$ into $A'$ is of the form
\[
  e'_j = e_j + \sum_{i=2}^k \alpha_j[i]\Parens{e_i-e_1}
\enspace,
\]
for some $\alpha_j[i]\in\Re$, i.e.,
\[
  e'_j[i] =
\begin{cases}
  -\sum_{i'=2}^k \alpha_j[i'] &\text{if $i=1$}
\\
  \alpha_j[i]                 &\text{if $i\in\set{2,\dotsc,k}$}
\\
  1                         &\text{if $i=j$}
\\
  0                         &\text{otherwise.}
\end{cases}
\]
The only solution of the above form that lies in $A'$ is obtained
by setting $\alpha_j[i]=-1/k$, yielding
\[
  e'_j[i] =
\begin{cases}
  1-1/k &\text{if $i=1$}
\\
  -1/k                 &\text{if $i\in\set{2,\dotsc,k}$}
\\
  1                         &\text{if $i=j$}
\\
  0                         &\text{otherwise.}
\end{cases}
\]
Therefore, for any pair of projections $e'_j$, $e'_{j'}$ for $j,j'>k$, and $j\ne j'$, we have
\[
  (e'_j-e_1)\inprod(e'_{j'}-e_1) = 1/k > 0
\enspace,
\]
so the projection of $\Omega$ is in an acute cone, i.e., acute angles hold.
\end{proof}

\begin{proof}[Proof of \Thm{log:partition}]
We begin by characterizing an $a_1$-perpendicular to a face $X\ne\Omega$ at $q$. Let
$\nu\coloneqq p(q)\not\in\aff(X)$, so the ambient space of the
perpendicular is $\aff(X\cup\set{\nu})$.
Thus, for a given $X$ and $q$, we will have the same $\im\gamma$ and
the same order on $\nu\in\im\gamma$ for any valid $a_1\in\tM$ which
allows to define the $a_1$-perpendicular to $X$ at $q$.
Recall that $P_q$ is the probability measure over $\Omega$ defined by
\[
  P_q(\omega)=e^{q\inprod\omega-C(q)}
\]
and note that $P_q(\omega)>0$ for all $\omega\in\Omega$.
Recall that
\[
  \nu = \sum_{\omega\in\Omega} P_q(\omega)\omega
\enspace.
\]
Let $X^c=\Omega\wo X$. Separate $\nu$ into components corresponding to $x\in X$ and $\omega\in X^c$:
\begin{align*}
  \nu_X &= \frac{1}{P_q(X)}\sum_{x\in X} P_q(x)x
\enspace,
\\
  \nu_{X^c} &= \frac{1}{P_q(X^c)}\sum_{\omega\in X^c} P_q(\omega)\omega
\enspace,
\end{align*}
i.e.,
\[
  \nu = P_q(X)\nu_X + P_q(X^c)\nu_{X^c}
\enspace.
\]
Since $\nu_X\in\aff(X)$, we have
\begin{equation}
\label{eq:affine:rewrite}
   \aff(X\cup\set{\nu})=\aff(X\cup\set{\nu_{X^c}})
\enspace.
\end{equation}
We will show that $\im\gamma$ consists exactly of the points
$(1-\talpha)\nu_X + \talpha\nu_{X^c}$ for $\talpha\in(0,1)$.

Consider $\nu'\in\im\gamma$ and $q'\in\pinv(\nu')\cap(q+X^\perp)$.
For any $x',x\in X$, we have
\[
  (q'-q)\inprod(x'-x) = 0
\enspace,
\]
so
\[
  q'\inprod(x'-x) = q\inprod(x'-x)
\enspace,
\]
and hence
\[
\frac{P_{q'}(x')}{P_{q'}(x)}
  = e^{q'\inprod(x'-x)}
  = e^{q\inprod(x'-x)}
  = \frac{P_{q}(x')}{P_{q}(x)}
\enspace.
\]
Since this holds for arbitrary $x,x'\in X$, we obtain
\begin{equation}
\label{eq:ratio}
\frac{P_{q'}(x)}{P_{q}(x)}
=
\frac{P_{q'}(X)}{P_{q}(X)}
\text{ for all }
x\in X
\enspace.
\end{equation}
Since $\nu'$ is in the ambient space of the perpendicular, which is
$\aff(X\cup\set{\nu})$, by \Eq{affine:rewrite}, we obtain
\[
  \nu'\in\aff(X\cup\set{\nu_{X^c}})
\enspace,
\]
so $\nu'$ can be written in the form
\begin{equation}
\label{eq:nu':decompose}
  \nu'=\sum_{x\in X}\alpha(x) x +\Parens{1-\sum_{x\in X} \alpha(x)}\nu_{X^c}
\end{equation}
for some $\alpha(x)\in\Re$ for $x\in X$. Also,
\[
  \nu'=\sum_{x\in X}P_{q'}(x) x +\sum_{\omega\in X^c}P_{q'}(\omega)\omega
\enspace.
\]
By the affine independence of $\Omega$, we therefore must have $\alpha(x)=P_{q'}(x)$.
Plugging this into \Eq{nu':decompose}, we obtain
\begin{align*}
  \nu' & =\Parens{\sum_{x\in X}P_{q'}(x) x} +P_{q'}(X^c)\nu_{X^c} \\
       & =\Parens{\frac{P_{q'}(X)}{P_q(X)}\sum_{x\in X}P_{q}(x) x} +P_{q'}(X^c)\nu_{X^c} \\
       & =P_{q'}(X)\nu_X + P_{q'}(X^c)\nu_{X^c}
\end{align*}
where the second equality follows by \Eq{ratio}. Thus, indeed
$\nu'=(1-\talpha)\nu_X + \talpha\nu_{X^c}$ for $\talpha=P_{q'}(X^c)$.

For any $I\subseteq[0,1]$, let
\[
  J_I\coloneqq \set{(1-\talpha)\nu_X + \talpha\nu_{X^c}:\:\talpha\in I}
\enspace.
\]
So far we have shown that $\im\gamma\subseteq J_{(0,1)}$; we next
argue that $J_{(0,1)}\subseteq\im\gamma$. We do this by
exploiting the continuity properties of $\gamma$.

For any $\talpha$, note that
\[
  \BigParens{(1-\talpha)\nu_X + \talpha\nu_{X^c}}
  \in
  \BigParens{\aff(X)+\talpha(\nu_{X^c}-\nu_X)}
  =A_\lambda
\]
for a suitable $\lambda$. Let $\talpha_0$ and $\lambda_0$
be the values associated with $\nu=p(q)$; note that both are
greater than zero. Then, for any $\lambda$,
the $\talpha$ such that
\[
  (1-\talpha)\nu_X + \talpha\nu_{X^c}\in A_\lambda
\]
is obtained as $\talpha(\lambda)\coloneqq\talpha_0\lambda/\lambda_0$.
The domain $\Lambda$ of
$\gamma$ is an open subset of the real line and therefore
it is a union of disjoint open intervals. Let $(\ulambda,\olambda)$
be the interval containing $\lambda_0$. For the sake of contradiction
assume that $\talpha(\ulambda)>0$.
Then $J_{[\talpha(\ulambda),\talpha_0]}$ is a compact subset of $\ri\M$,
and therefore by \Thm{perp:cont}, its intersection with $\im\gamma$ should
be compact as well. However, by the selection of $\ulambda$, the intersection
equals $J_{(\talpha(\ulambda),\talpha_0]}$, which is not compact, yielding
a contradiction. Similarly, we obtain a contradiction if $\talpha(\olambda)<1$.
Therefore, $\talpha(\ulambda)\le 0$ and $\talpha(\olambda)\ge 1$, showing that
$J_{(0,1)}\subseteq\im\gamma$, and hence in fact $J_{(0,1)}=\im\gamma$.

Now we are ready to prove the theorem. Let $\nu'\in\im\gamma$ such that $\nu'\succeq\nu$, and let $q\in\pinv(\nu)$, $q'\in\pinv(\nu')$. Use the
notation $P\coloneqq P_q$
and $P'\coloneqq P_{q'}$, and write
\begin{align*}
  (q'-q)\inprod(\omega-x)
&{}=\ln\Parens{\frac{e^{q'\inprod\omega}}{e^{q'\inprod x}}
            \cdot
            \frac{e^{q\inprod x}}{e^{q\inprod\omega}}} \\
&{}=\ln\Parens{\frac{P'(\omega)}{P'(x)}
            \cdot
            \frac{P(x)}{P(\omega)}} \\
&{}=\ln\Parens{\frac{P'(\omega)}{P(\omega)}
            \cdot
            \frac{P(x)}{P'(x)}}
\enspace.
\end{align*}
For $\omega\in X^c$, our characterization
of the perpendicular implies that $P(x)\ge P'(x)$ and $P(\omega)\le P(\omega')$
since $\nu'$ has a larger (or equal) coefficient $\talpha$ than $\nu$,
because it is further (or equally) away from $\nu_X$. Thus, the above expression
is non-negative, yielding the acute angles property.
\end{proof}

\begin{proof}[Proof of \Thm{1d}]
Let $\mu\in\tM$, $q\in\pinv(\tM)$ and let $\im\gamma$ be the $\mu$-perpendicular to $X$ at $q$.
Note that the perpendicular is well defined only if $X$ is a singleton, say $X=\set{x}$,
and $\mu\ne x$. Let $\set{x'}$ be the other singleton face of $\M$.
Thus, $\im\gamma=\aff(\set{x,x'})\cap(\ri\dom C^*)$ with the direction from
$x$ towards $x'$.
Note that $\witness(X)=\set{u:\:u\cdot(x'-x)\ge 0}$.
Let $\nu=p(q)$ and let $\nu'\succeq\nu$, i.e., $\nu'-\nu=\lambda(x'-x)$ for some $\lambda>0$.
Pick $q'\in\pinv(\nu')$, which exists, because $\nu'\in\ri\dom C^*$.
By convexity, we have
\begin{align*}
 0 & \le(q'-q)\cdot(\nabla C(q')-\nabla C(q))=(q'-q)\cdot(\nu'-\nu) \\
 & =\lambda(q'-q)\cdot(x'-x)
\enspace,
\end{align*}
i.e., $q'-q\in\witness(X)$.
\end{proof}

\begin{proof}[Proof of \Thm{direct:sum}]
Let $\M_1=\conv(\Omega_1)$ and $\M_2=(\conv\Omega_2)$.
We first argue that $\M=\conv(\Omega_1\times\Omega_2)=(\conv\Omega_1)\times(\conv\Omega_2)=\M_1\times\M_2$. For $i\in\set{1,2}$, let $\nu_i\in\M_i$, i.e.,
for some probability measure $P_i$ on $\Omega_i$, we have $\E_{\omega_i\sim P_i}[\omega_i]=\nu_i$. Defining the probability measure $P$ on $\Omega$ by
$P(\omega_1,\omega_2)=P_1(\omega_1)P_2(\omega_2)$, we obtain $\E_{(\omega_1,\omega_2)\sim P}[(\omega_1,\omega_2)]=(\nu_1,\nu_2)$,
i.e., $(\nu_1,\nu_2)\in\M$. Conversely, let $(\nu_1,\nu_2)\in\M$, i.e., for some measure $P$ on $\Omega$, $\E_{(\omega_1,\omega_2)\sim P}[(\omega_1,\omega_2)]=(\nu_1,\nu_2)$.
But this also means that $\E_{(\omega_1,\omega_2)\sim P}[\omega_i]=\nu_i$ for $i\in\set{1,2}$, so $\nu_i\in\M_i$ for $i\in\set{1,2}$. Thus,
$\M=\M_1\times\M_2$.

We also have $\dom C=(\dom C_1)\times(\dom C_2)$ and $\ri\dom C=(\ri\dom C_1)\times(\ri\dom C_2)$, which implies that $\tM\subseteq\tM_1\times\tM_2$.

We next show that $X$ is a face of $\M$ if and only if $X=X_1\times X_2$ where $X_1$ is a face of $\M_1$ and $X_2$ is a face of $\M_2$.
A face $X$ of $\M$ is characterized by a vector $u$ and a scalar $c$ such that
\begin{align*}
  u\cdot x &= c
  \text{ for }x\in X
\\
  u\cdot\omega &>c
  \text{ for }\omega\in\Omega\wo X
\enspace.
\end{align*}
If $X_1$ is a face of $\M_1$ characterized by $u_1$ and $c_1$, and $X_2$ is a face of $\M_2$ characterized by $u_2$ and $c_2$, then we immediately obtain
that $X_1\times X_2$ is a face of $\M$ characterized by $u=(u_1,u_2)$ and $c=c_1+c_2$. Conversely, assume $X$ is a face of $\M$ characterized by $u=(u_1,u_2)$ and $c$.
We first show that $X$ is a Cartesian product. We proceed by contradiction and assume that $(x_1,x_2)\in X$ and $(x'_1,x'_2)\in X$, but $(x_1,x'_2)\not\in X$. By assumption:
\begin{align*}
  u_1\cdot x_1 + u_2\cdot x_2 &= c
\\
  u_1\cdot x'_1 + u_2\cdot x'_2 &= c
\\
  -u_1\cdot x_1 - u_2\cdot x'_2 &< -c
\intertext{%
Summing the above three yields:}
  u_1\cdot x'_1 + u_2\cdot x_2 &< c
\end{align*}
which is a contradiction with $X$ being a face. By symmetry, we also obtain $(x'_1,x_2)\in X$, hence $X=X_1\times X_2$ for some $X_1\subseteq\Omega_1$ and $X_2\subseteq\Omega_2$.
Let $(x_1,x_2)$ be some element of $X$. Note that $X_1$ must be a face of $\M_1$ characterized by $u_1$ and $c_1\coloneqq u_1\cdot x_1$, because
otherwise, there would exist $\omega_1\in\Omega_1\wo X_1$ such that $u_1\cdot\omega_1\le c_1=u_1\cdot x_1$, i.e.,
\[
  u_1\cdot\omega_1+u_2\cdot x_2 \le u_1\cdot x_1 + u_2\cdot x_2
\]
which would contradict $X$ being a face, since $(\omega_1,x_2)\not\in X$. By symmetry, we also obtain that $X_2$ is a face of $\M_2$.

Let $\gamma$
be the $a$-perpendicular to $X$ at $q$, where $a=(a_1,a_2)$, $X=X_1\times X_2$ and $q=(q_1,q_2)$. Note that $a\in\tM$, $p(q)\in\tM$ implies
that $a_i\in\tM_i$, $p_i(q_i)\in\tM_i$ because $\tM\subseteq\tM_1\times\tM_2$.
The optimization of $D(q,\nu)$ over $A_\lambda$ in the definition of the perpendicular
decomposes into independent convex problems in $\nu_1$ and $\nu_2$,
because $\aff(X)=\aff(X_1)\times\aff(X_2)$ and $D(q,\nu)=D_1(q_1,\nu_1)+D_2(q_2,\nu_2)$, where $D_1$ and $D_2$ are the
divergences derived from $C_1$ and $C_2$. Thus, for any point $\nu'\in\im\gamma$ such that $\nu'\succeq p(q)$,
we obtain that for $i\in\set{1,2}$ the components $\nu'_i$ lie
on the $a_i$-perpendicular
to $X_i$ at $q_i$ and $\nu'_i\succeq p_i(q_i)$.
Let $x=(x_1,x_2)\in X$ and $\omega=(\omega_1,\omega_2)\in\Omega$.
Then by acute angles assumption, we can choose $q'_i\in\pinv_i(\nu'_i)\cap[q_i+\witness(X_i)]$. Let $q'=(q'_1,q'_2)$.
Note that $q'\in\pinv(\nu')$. We will argue that also $(q'-q)\in\witness(X)$:
\[
   (q'-q)\inprod(\omega-x)=\sum_{i\in\set{1,2}} (q'_i-q_i)\inprod(\omega_i-x_i)\ge 0
\enspace,
\]
where the last inequality follows, because $(q'_i-q_i)\in\witness(X_i)$. Thus,
the acute angles assumption holds for $C$ and $\Omega$.
\end{proof}

\section{PROOF OF THEOREM 5.2}
\label{app:proof}

In this section we give the complete proof of \Thm{main},
with the revised definition of budget additivity, which requires that $q_0$
be arbitrage-free
(see \App{arb:free}).
The proof proceeds in several steps. Let $\nu_0=p(q_0)$.
Assuming acute angles, we begin by constructing an oriented curve $L$
joining $\nu_0$ with $\mu$, by sequentially choosing portions of perpendiculars
for monotonically decreasing active sets. We then show
that budget additivity holds for any solutions with prices in $L$,
and finally show that the curve $L$
is the locus of the optimal prices of solutions $\hQ(q_0)$, as
well as optimal prices of solutions $\hQ(q)$ for any $q\in\hQ(q_0)$.

\paragraph{Part 1: Construction of the solution path $L$}
In this part, we construct:
\begin{itemize}
\item a sequence of prices $\nu_0,\nu_1,\dotsc,\nu_k$ with $\nu_0=p(q_0)$ and $\nu_k=\mu$
\item a sequence of oriented curves $\ell_0,\dotsc,\ell_{k-1}$ where each $\ell_i$ goes from $\nu_i$ to $\nu_{i+1}$
\item a monotone sequence of sets $\Omega\supseteq X_0\supset X_1\supset\dotsb\supset X_k=\emptyset$, such that
      the following \emph{minimality} property holds: $X_i$ is the minimal face for
      all $\nu\in(\im\ell_i)\wo\set{\nu_{i+1}}$ for $i\le k-1$,
      and $X_k$ is the minimal face for $\nu_k$.
\item a sequence of states $q_1,\dotsc,q_{k-1}$ such that $q_i\in\pinv(\nu_i)\cap[q_{i-1}+\witness(X_{i-1})]$
\end{itemize}
The curves $\ell_i$ will be referred to as \emph{segments}.
The curve obtained by concatenating the segments $\ell_0$ through $\ell_{k-1}$ will
be called the \emph{solution path} and denoted $L$. In the special case
that $\nu_0=\mu$, we have $k=0$, $X_0=\emptyset$ and $L$ is a degenerate curve
with $\im L=\set{\mu}$.

If $\nu_0\ne\mu$, we construct the sequence of segments iteratively.
Let $X_0\ne\emptyset$ be the minimal face such that $\nu_0\in \conv(X_0\cup\set{\mu})$.
By the minimality, $\mu\not\in\aff(X_0)$.
Let $\gamma$ be the $\mu$-perpendicular to $\aff(X_0)$ at $q_0$. The curve $\gamma$ passes through $\nu_0$ and eventually reaches the boundary of $\conv(X_0\cup\set{\mu})$ at some $\nu_1$ by continuity of $\gamma$ (see \Thm{perp:cont}). Let segment $\ell_0$ be the portion of $\gamma$ going from $\nu_0$ to $\nu_1$.

This construction gives us the first segment $\ell_0$. There are two possibilities:
\begin{enumerate}
\item $\nu_1=\mu$; in this case we are done;
\item $\nu_1$ lies on a lower-dimensional face of $\conv(X_0\cup\set{\mu})$; in this case, we pick
    some $q_1\in\pinv(\nu_1)\cap [q_0+\witness(X_0)]$, which can be done by the acute angles assumption, and use the above construction again, starting with $q_1$, and obtaining a new set $X_1 \subset X_0$ and a new segment $\ell_1$; and iterate.
\end{enumerate}
The above process eventually ends, because with each iteration, the size of the active set decreases.
This construction yields monotonicity of $X_i$ and the minimality property.

The above construction yields a specific sequence of $q_i\in\pinv(\nu_i)\cap[q_{i-1}+\witness(X_{i-1})]$. We will now show that actually $q_i\in\pinv(\nu_i)\cap(q_0+X_{i-1}^\perp)$ and that the construction of $L$ is
independent of the specific $q_1,q_2,\dotsc,q_{k-1}$ chosen. To begin, note that from our construction, we can write $q_i=q_0+u_0+u_1+\dotsc+u_{i-1}$ for some $u_j\in\witness(X_j)\subseteq X_j^\perp$. Since $X_i\subseteq X_j$ for $j=1,\dotsc,i-1$, we actually have $u_j\in X_{i-1}^\perp$, so $q_i\in (q_0+X_{i-1}^\perp)$.
Note
that $X_{i-1}^\perp\subseteq X_i^\perp$, and according to \Prop{perp:equiv},
any $q_i\in\pinv(\nu_i)\cap(q_0+X_i^\perp)$ yields the same $\mu$-perpendicular to
$\aff(X_i)$ and hence the same segment $\ell_i$. By induction it therefore follows that the segments
$\ell_0,\dotsc,\ell_{k-1}$ are uniquely determined by our construction regardless of the specific $q_1,\dotsc,q_{k-1}$.

\paragraph{Part 1': The solution path starting at a midpoint} Let $\nu\in\im L$,
and $q\in\pinv(\nu)\cap(q_0+X_\nu^\perp)$, and let $L'$ be the solution path if the initial
state were $q$ rather than $q_0$. By a similar reasoning as in the previous paragraph, we see that
$L'$ is a restriction of $L$ starting with $\nu$.

\paragraph{Part 2: Budget additivity for points on $L$}
Let $\nu,\nu'\in\im L$ such that $\nu\preceq\nu'$. Let $q\in\hQ(\nu;q_0)$
and $q'\in\hQ(\nu';q)$ such that $q\in\hQ(B;q_0)$ and $q'\in\hQ(B';q)$. Note
that by \Prop{init:hQ}, $q$ is arbitrage-free.
In this part we show that $q'\in\hQ(B+B';q_0)$.

First, consider the case that $\nu'=\mu$. To see that $q'\in\hQ(B+B';q_0)$, first note that the constraints of
\CP\ hold, because
$
    U(q',\omega;q_0)=U(q',\omega;q)+U(q,\omega;q_0)\ge -B'-B
$
for all $\omega$
by path independence of the utility function. As noted in the introduction, in the absence
of constraints, the utility $U(\bq,\mu;q_0)$ is maximized at any $\bq$ with $p(\bq)=\mu$. Thus,
$q'$ is a global maximizer of the utility and satisfies the constraints, so $q'\in\hQ(B+B';q_0)$.
If $\nu=\mu$, we must also have $\nu'=\mu$ and the statement holds by previous reasoning.

In the remainder, we only analyze the case $\nu\preceq\nu'\prec\mu$.
This means that $\nu\in(\im\ell_i)\wo\set{\nu_{i+1}}$ and $\nu'\in(\im\ell_j)\wo\set{\nu_{j+1}}$ for $i\le j$. By \Thm{KKT:final}, we therefore
must have $q\in[q_0+\witness(X_i)]$ and $q'\in[q+\witness(X_j)]$. By anti-monotonicity
of witness cones, $\witness(X_j)\supseteq\witness(X_i)$ and hence, $q'\in[q_0+\witness(X_j)]$, yielding
\[
   q'\in\hQ(\nu';q_0)
\enspace.
\]

We now argue that the budgets add up.
Let $x\in X_j\subseteq X_i$. By \Lem{KKT:initial}, we obtain that
\begin{align*}
&q\in\hQ(B;q_0)\text{ for }
   B=-U(q,x;q_0)
\enspace,
\\
&q'\in\hQ(B';q)\text{ for }
   B'=-U(q',x;q)
\enspace,
\\
&q'\in\hQ(\bB;q_0)\text{ for }
   \bB=-U(q',x;q_0)
\enspace.
\end{align*}
However, by path independence of the utility function
\[
   \bB=-U(q',x;q_0)
      =-U(q',x;q)-U(q,x;q_0)
      =B'+B.
\]

\paragraph{Part 3: $L$ as the locus of all solutions}
In this part we show that
\[
    \hQ(q_0) = \bigcup_{\nu\in\im L} \hQ(\nu;q_0)
\enspace.
\]

We will
begin by defining sets of budgets for which the optimal price is $\nu$ and show
that their union across all $\nu\in\im L$ is a closed interval. Since both $\nu_0,\mu\in\im L$,
this will mean that we have included price vectors across all possible budgets. The statement
of Part 3 will then follow by \Thm{budget:constraints}.

Let $x\in X_{k-1}=\bigcap_{i=0}^{k-1} X_i$
and let $B(q)\coloneqq -U(q,x;q_0)$.
Further, for $\nu\in\im\ell_i$, let
\[
   B_i(\nu)\coloneqq\set{B(q):\:q\in\pinv(\nu)\cap[q_0+\witness(X_i)]}
\enspace.
\]
From \Cor{budget:nu}, we know that for $\nu\in(\im\ell_i)\wo\set{\nu_{i+1}}$, $B_i(\nu)$ is exactly
the set of budgets for which $\nu$ is the optimal price vector. The set
$B_i(\nu_{i+1})$ is potentially only a subset of such budgets (corresponding to $X_i$
being the tight set, rather than the actual minimal set $X_{i+1}$). First we show
that $B_i(\nu)$ is non-empty for $\nu\in\im\ell_i$. Let $\nu\in\im\ell_i$. By acute angles assumption,
there exists $q\in\pinv(\nu)\cap[q_i+\witness(X_i)]$. Furthermore, $q_j\in[q_{j-1}+\witness(X_{j-1})]$
for $j=1,\dotsc,i$, so we can write $q=q_0+u_0+\dotsb+u_i$ where $u_j\in\witness(X_j)$.
By anti-monotonicity of witness cones, $\witness(X_j)\subseteq\witness(X_i)$ for
$j=1,\dotsc,i$, so we actually have $u_j\in\witness(X_i)$ and thus $q\in[q_0+\witness(X_i)]$,
proving that the set $B_i(\nu)$ is non-empty.

We will next show that
\[
  B_i(\ell_i)\coloneqq\bigcup_{\nu\in\im\ell_i} B_i(\nu)
\]
is an interval.

Consider a fixed $\nu\in\im\ell_i$. For $q\in\pinv(\nu)$,
we have $C(q)=q\inprod\nu-C^*(\nu)$, i.e., $B(q)$ is linear in $q$ over $q\in\pinv(\nu)$.
Since the set $\pinv(\nu)$ is closed and convex,
so is $\pinv(\nu)\cap[q_0+\witness(X_i)]$. The latter set is also non-empty,
hence the set $B_i(\nu)$
must be a non-empty closed interval. Let $B_i^{\min}(\nu)$ and $B_i^{\max}(\nu)$ be the lower and upper
endpoints of $B_i(\nu)$.
Since the budget additivity holds along $L$ (by Part 2),
we must have that $B_i^{\max}$ is non-decreasing on $\ell_i$.
Next note that for $\nu\ne\nu'$ the sets $B_i(\nu)$ and $B_i(\nu')$ must be
disjoint. This implies that $B_i^{\max}$ is actually increasing and so
is $B_i^{\min}$.

We next show that $B_i^{\max}$ is right-continuous on $\ell_i$.
Let $M\coloneqq\set{(\nu,q):\:\nu\in\im\ell_i,\,q\in\pinv(\nu)}$.
By
\Thm{perp:cont}, the set $M$ can be written $\setC+(0\oplus\setL^\perp)$ where $\setC$ is
compact. Let $M'\coloneqq\set{(\nu,q):\:\nu\in\im\ell_i,\,q\in\pinv(\nu)\cap[q_0+\witness(X_i)]}$. Since the
set $[q_0+\witness(X_i)]$ is closed, the set $M'$ can be written as $\setC'+(0\oplus\setL^\perp)$
where $\setC'$ is compact.
To show that $B_i^{\max}$ is right-continuous,
pick $\nu\in\im\ell_i$ and let $\set{\nu'_t}_{t=1}^\infty$ be a sequence of $\nu'_t\in\ell_i$, $\nu'_t\succeq\nu$
such that $\lim_{t\to\infty}\nu'_t=\nu$. Pick $q'_t$ such that $(\nu'_t,q'_t)\in\setC'$ and $B(q'_t)=B_i^{\max}(\nu'_t)$.
By compactness, the sequence $\set{(\nu'_t,q'_t)}_{t=1}^\infty$ must have a cluster point $(\nu,q)\in\setC'$, and by continuity of $B$, we must
have $\lim_{t\to\infty} B_i^{\max}(\nu'_t)=\lim_{t\to\infty} B(q'_t)=B(q)\le B_i^{\max}(\nu)$. The
right continuity of $B_i^{\max}$ now follows by monotonicity. By symmetric reasoning, $B_i^{\min}$ must be left-continuous.

Now for the sake of contradiction, assume that $B_i(\ell_i)$ is not an interval,
i.e., assume that there is a value $B^*\not\in B_i(\ell_i)$ such that some higher and lower values are in $B_i(\ell_i)$. By monotonicity, there must exist $\nu^*$ such that $B_i^{\max}(\nu)<B^*$ for $\nu\prec\nu^*$ and $B_i^{\min}(\nu)>B^*$ for $\nu\succ\nu^*$. However, this means that
\[
    B_i^{\min}(\nu^*) =   \lim_{\nu\uparrow\nu^*}B_i^{\min}(\nu^*)
                      \le \lim_{\nu\uparrow\nu^*}B_i^{\max}(\nu^*)
                      \le B^*
\]
and
\[
    B_i^{\max}(\nu^*) =   \lim_{\nu\downarrow\nu^*}B_i^{\max}(\nu^*)
                      \ge \lim_{\nu\downarrow\nu^*}B_i^{\min}(\nu^*)
                      \ge B^*
\]
which means that $B^*\in B_i(\nu^*)$ yielding a contradiction.

Finally, note that $B_i(\nu_{i+1})\subseteq B_{i+1}(\nu_{i+1})$ for $i\le k-1$, hence
$\bigcup_{i=1}^{k-1} B_i(\ell_i)$ is an interval as well.

\paragraph{Part 3': $L$ as the locus of solutions starting at a midpoint}
Let $\nu\in\im L$ and $q\in\hQ(\nu;q_0)$. Since $\hQ(\nu;q_0)\subseteq\pinv(\nu)\cap(q_0+X_\nu^\perp)$,
Part 1' yields that the solution path $L'$ for $q$ coincides with the portion of
$L$ starting at $\nu$. By \Prop{init:hQ}, $q$ is
arbitrage-free, so the reasoning of the previous part can be applied
to $L'$, yielding the following statement:
\[
    \hQ(q) = \bigcup_{\substack{\nu'\in\im L:\\\nu'\succeq\nu}} \hQ(\nu';q)
\enspace.
\]

\paragraph{Part 4: Proof of the theorem}
Let $B,B'\ge0$ and $q\in\hQ(B; q_0)$ and $q'\in\hQ(B';q)$. From Parts 3 and 3', we know that $q\in\hQ(\nu;q_0)$ and $q'\in\hQ(\nu';q)$ for some $\nu,\nu'\in\im L$ such that $\nu\preceq\nu'$.
By Part 2, we therefore obtain that $q'\in\hQ(B+B';q_0)$, proving the theorem.

\end{document}